\documentclass[journal]{IEEEtran}
\usepackage{xcolor}
\definecolor{MXY}{RGB}{252,8,235}
\usepackage{cite}
\usepackage{amsmath,amssymb,amsfonts}
\usepackage{amsmath}
\usepackage{booktabs,multirow}
\usepackage{amsthm} % Mathematical theorem environments

\usepackage{algorithm}
\usepackage{graphicx}
\usepackage{textcomp}
\usepackage{subcaption} 
\usepackage{algpseudocode}
\usepackage[colorlinks=true, allcolors=blue]{hyperref}
\usepackage{xcolor}
\usepackage{orcidlink}
\usepackage{optidef}

\newtheorem{theorem}{Theorem}

\newtheorem{proposition}{Proposition}

\newtheorem{remark}{Remark}
\DeclareMathOperator{\diag}{diag}

\begin{document}
\bstctlcite{IEEEexample:BSTcontrol}

\title{Optimal Multi-RIS Placement: Coverage-Guaranteed Sum Rate Maximization Under Inhomogeneous User Distributions}
%Graph Theory Based , Multi-RIS Placement with Performance Maximization and Bayesian Optimization

\author{\IEEEauthorblockN{Abhishek Rajasekaran\orcidlink{0009-0006-0455-8492}, \textit{Member, IEEE}, Mehdi Karbalayghareh\orcidlink{0000-0002-6857-5914}, \textit{Member, IEEE},\\ Xiaoyan Ma \orcidlink{0000-0002-5287-4215}, \textit{Member, IEEE}, Sunjung Kang\orcidlink{0000-0002-0573-720X}, \textit{Member, IEEE},\\
David J. Love\orcidlink{0000-0001-5922-4787}, \textit{Fellow, IEEE}, and Christopher G. Brinton\orcidlink{0000-0003-2771-3521}, \textit{Senior Member, IEEE}}
\thanks{A. Rajasekaran, M. Karbalayghareh, X. Ma, S. Kang, D. J. Love, and C. G. Brinton are with Electrical and Computer Engineering, Purdue University, IN 47907, USA (e-mail: \href{mailto:rajasek1@purdue.edu}{rajasek1@purdue.edu}; 
 \href{mailto:mkarbala@purdue.edu}{mkarbala@purdue.edu}; \href{mailto:ma946@purdue.edu}{ma946@purdue.edu};
 \href{mailto:kang392@purdue.edu}{kang392@purdue.edu};
\href{mailto:djlove@purdue.edu}{djlove@purdue.edu};
\href{mailto:cgb@purdue.edu}{cgb@purdue.edu}). This work was supported in part by the National Science Foundation (NSF) under grants ITE-2326898, EEC-1941529, CNS-2212565, CNS-2146171,  CNS-2225578 and by the Office of Naval Research (ONR) under grant N00014-21-1-2472.} \thanks{A preliminary version of this work has been accepted to the 2026 WiOpt Conference~\cite{conf}.}} %
\markboth{IEEE Transactions on Wireless Communications}%
{A. Rajasekaran \MakeLowercase{\textit{et al.}} }
\maketitle

\begin{abstract}
Reconfigurable Intelligent Surface (RIS) has emerged as a promising next-generation technology that improves the throughput and coverage of a wireless system. The realization of the full potential of RISs in a wireless system is tied to their strategic spatial deployment. While existing literature on RIS placement primarily focuses on maximizing coverage, when multiple RIS placements guarantee the required coverage (happens quite often), these approaches fail to exploit prior user trends to choose the one that is most probable to maximize throughput. Thus, to enable throughput maximization while guaranteeing fairness, we formulate a novel hierarchical problem that maximizes the expected sum rate of the system while guaranteeing a certain probabilistic coverage, with the requisite minimum number of RISs deployed. To solve this multi-layered non-convex problem, firstly, we obtain a set of optimal points where we can deploy RISs to provide the coverage guarantee. Then, the least number of RISs that can guarantee the required coverage is obtained by a greedy minimum partitioning. Finally, a Bayesian optimization based approach is used to compute the optimal RIS placement. Numerical results are provided to show that the proposed framework consistently identifies placements that jointly achieve good coverage and throughput, without impractical assumptions.
\end{abstract}
%While existing literature primarily focuses on determining optimal RIS placement by enabling fairness i.e., maximizing coverage to navigate through obstacles, these approaches often fail to exploit the spatial distribution of user density to maximize throughput.
%, Joint beamforming

% \fi
\IEEEpeerreviewmaketitle

\section{Introduction}
\label{sec:intro}
\IEEEPARstart{N}{ext}-generation wireless as envisioned by major countries and companies~\cite{nextg2022roadmap,uusitalo2021hexa,fiercewireless2020iab} cater to a wide variety of applications including Immersive Extended Reality (XR), Tactile internet, Holographic Communication, Ubiquitous Intelligence, Large Telecom Models (LTMs), etc. Many of these necessitate an increase in data rates and an improvement in the coverage of the base station (BS). These requirements are reflected in the ITU-R Recommendation M.2160-0 for 6G (IMT-2030) which explicitly lists immersive communication (as an advancement of enhanced Mobile Broadband (eMBB)), Massive Communication and Ubiquitous Connectivity as usage scenarios~\cite{ITUR}. Reconfigurable Intelligent Surface (RIS) has been well-established to be capable of increasing the throughput and coverage of wireless systems~\cite{RIS3,coverage1}.

RIS is typically a planar surface with passive reflecting elements made of a metamaterial and is connected to a smart controller~\cite{RIS2,pathloss}. The nature of the metamaterial is such that it can help alter the amplitude and/or phase shifts of the reflected signal in an intelligent manner, enabling RIS to convincingly emerge as a tool to improve the throughput of the system~\cite{RIS1}. Due to the generated multipath, RIS can also enhance coverage by overcoming blockages created by obstacles~\cite{coverage2,coverage3}--a capability that is very valuable in next-generation wireless which is proposed to use higher-carrier-frequencies~\cite{ITUR}, exhibiting severe blockage issues~\cite{6glast,6gbrinton}. 
%Unlike conventional amplify-and-forward relays, these are very cheap and flexible; hence, they can be easily deployed in large numbers if needed~\cite{relay}.Thus, with estimated channel information, one can generate phases as required using the controller to ensure that the desired signals combine constructively and interferers combine destructively at the receiver, enhancing system performance~\cite{risone,ristwo}.

Given that enhancing coverage and throughput are the major uses of an RIS, its optimal placement in the network is the most crucial step that enables these enhancements~\cite{locbasis1,locbasis2}. This problem is challenging for many reasons including the practical conundrum that optimal RIS locations must be determined prior-deployment in order to maximize post-deployment coverage and throughput. Also, once the system is deployed with RISs at their optimal locations, then it is necessary to design optimal transmit beamformers and RIS phase shifts to maximize the throughput of an RIS-assisted system in real-time. Both these problems -- placement and beamforming have been investigated previously.

% \textcolor{blue}{Then before "related works", you need to add 1 (or 2) very important paragraph explaining why this problem is important and what the challenges are.} 

\subsection{Related Works}%An outage probability based approach for RIS placement that results in a similar close form expression is~\cite{locandpower}.

Many existing works that discuss RIS placement use coverage maximization approaches listed in~\cite{loc1,loc3}. These works propose to use geometric approaches to help the desired signals navigate around obstacles by positioning RIS(s) at relevant positions. However, when these algorithms result in multiple optimal RIS configurations that yield the required coverage -- a frequent occurrence, these algorithms lack the capability to identify that configuration that is most likely to maximize the throughput of the system. On the other hand, algorithms that obtain optimal RIS placement by maximizing the real-time throughput of the system like~\cite{loc2} (jointly with the transmit beamformers, RIS phase coefficients, and other variables) are not practical since they require the RIS location to change in real-time. Also, these approaches result in simplified purely-geometry dependent closed form solutions and work only for fixed number of users and fixed user locations. Certain works including~\cite{locapp1,locapp4,locapp3,locvisapp1,locvisapp2} have also discussed optimization of RIS placement for specific applications like Unmanned Aerial Vehicles (UAV), high-speed trains, Visual Light Communication (VLC) and autonomous vehicles. Because the objectives of these approaches are tailored to specific requirements of each application, none of them offer the generality of the proposed method, with universal applicability.   %It is shown later that we can do much better than such purely-distance dependent closed form solutions for RIS placement.%locapp2,

The other important problem, throughput maximization, is typically dealt with by jointly optimizing the transmit beamformer and RIS coefficients in real-time. We define this problem as the joint beamforming problem. By definition, to perform joint beamforming for the current `channel state', we need to obtain the corresponding channel state information (CSI). To address this, \cite{channelesti2, channelesti1,channelestisurvey1,channelestisurvey2} present various channel estimation approaches for RIS-assisted systems. These methods can be extended for the multi-RIS case by sequentially turning off and on each of the RISs. With estimated CSI available using one of these approaches, one can then perform joint beamforming to maximize the throughput of the system.~\cite{beam0,weightedsumrate,beam2, loc2} present various approaches of solving the joint beamforming problem; most common one using fractional programming based approach to maximize the throughput of the RIS-assisted wireless system. Though, in our work, we illustrate channel estimation and joint beamforming using one of the aforementioned methods, the proposed RIS placement algorithm remains effective for any channel estimation and/or joint beamforming technique as required. For instance, a recent work uses Orthogonal Complement Projection for channel estimation of RIS-assisted systems to reduce the overhead~\cite{channelestimulti}.  

In summary, based on the survey, we see that although the RIS-assisted system is well-studied, there is a practically relevant problem that remains unresolved. The existing coverage-based RIS placement approaches ensure fairness by reaching all users, but aren't capable of maximizing throughput (sum rate) and the sum rate maximization approaches require RIS location to change in real-time which isn't practical. Therefore, there is a clear scope to propose an RIS placement technique that does both -- \textit{to maximize expected system throughput while ensuring a level of fairness, making use of minimum resources.} 

\subsection{Contributions}
Motivated by the research gap, we propose a novel scheme to maximize the expected sum rate of the system with a probabilistic coverage guarantee, while using minimum number of RISs. The key contributions include:
\begin{itemize}
    \item Firstly, a detailed justification is provided to show that a simple coverage maximization or a sum rate maximization is an oversimplification. Then, a unique hierarchical softly-coupled formulation is proposed to better deal with the placement problem. More specifically, we propose to first compute the minimum number of RISs and a search space to provide a probabilistic coverage guarantee. Then, we propose to optimize RIS placement by maximizing the expected sum rate performance on this space. This formulation, to the best of our knowledge, is the first attempt to simultaneously consider expected throughput, probabilistic coverage and number of RISs.
    \item To solve for the search space where RISs can be placed to satisfy the probabilistic coverage requirement, firstly, we propose and prove a theorem to reduce this global probabilistic coverage guarantee to simpler local conditions that can be applied to the regions shadowed by each of the obstacles. To reduce the number of computations, instead of handling each obstacle and the areas they shadow individually, we propose to form clusters among obstacles. Then, a proposition is made and proved to define candidate sets for RIS placement that satisfy the derived local coverage requirements. Finally, we calculate these candidate sets using a visibility graph.% on this visibility graph.
    \item Then, to solve for the minimum number of RISs we require to satisfy the coverage requirement, we reduce the redundancy in the candidate sets defined. We construct an intersection hypergraph~\cite{hypergraph} using these candidate sets. Taking inspiration from~\cite{GMS}, we propose a greedy maximal hypergraph partitioning algorithm to reduce the number of candidate sets to the lowest possible. %Thus, we get the minimum number of RISs and also compute the corresponding reduced search space for RIS placement.
    \item Finally, we propose a novel technique that combines user distribution-driven learning and conventional joint beamforming techniques to calculate the optimal RIS placement. We use a Bayesian optimization (BO) based approach to maximize the expected sum rate ~\cite{bayesopt}.
    % \item Furthermore, the numerical results presented demonstrate that the proposed scheme simultaneously maximizes coverage by overcoming the obstacle configuration and optimizes sum rate performance by exploiting the spatial non-homogeneity of the user distribution. Results show that the proposed framework consistently identifies a point that jointly achieves strong coverage and strong performance, without requiring impractical system assumptions. To the best of our knowledge, this is the first work to achieve both objectives simultaneously in a practically relevant setting.
    % %Critically, the proposed hierarchical formulation avoids the cons of both extreme approaches — a purely coverage-driven approach that sacrifices performance by over-prioritizing complete coverage, and a purely performance-driven approach that risks leaving certain areas unserved.
\end{itemize}
\cite{conf} is our preliminary work that introduces this idea. In this work, we address its limitations with a general formulation for multiple RISs that provides explicit control over the throughput-coverage tradeoff and additionally incorporates resource minimization, search space determination that considerably reduces computational cost, and a more efficient BO.  %, it had three key limitations: it offered limited control over the coverage-throughput tradeoff, it lacked a resource minimization component, and it lacked a search space restriction procedure for reducing computational cost.

% A preliminary version of this work was presented for a single RIS case~\cite{conf}.

% \textcolor{blue}{MK: the contribution bullet points can be summarized, especially 1 and 2. }

\subsection{Organization and Notations}
The rest of the paper is organized as follows. Section \ref{sec:sysmodel} sets up the system model for the proposed problem. Then, Section \ref{sec:prblmform} provides a detailed discussion on the problem formulation. In Section \ref{sec:ss}, the computation of the search space for RIS placement and optimization of the number of RISs are discussed. In Section \ref{sec:perf}, the proposed novel method to maximize the expected system throughput is discussed. Then, Section \ref{sec:results} lays down the detailed numerical and analytical findings of the work. Finally, Section \ref{sec:conc} presents the conclusion.

The notations used in this paper as listed as follows. All scalar valued variables are listed normally (e.g.: $a$) whereas vectors and matrices are listed in bold (e.g.: $\boldsymbol{a}$). Throughout this work, braces $\{\cdot\}$ would denote a mathematical set of values. $|\cdot|$ and $\mathcal{R}e(\cdot)$ denotes the absolute value and real part of the enclosed value respectively. All norms in this paper are assumed to be L2 norms. $(\cdot)^T$, $(\cdot)^*$, $(\cdot)^\mathrm{H}$ indicate the transpose, conjugate, and conjugate transpose of the enclosed matrix, respectively. $\boldsymbol{I}_N$ indicates an identity matrix of size $N$. $\boldsymbol{1}_N$ denotes a column vector with all 1s of size $N$. $\diag(\cdot)$ denotes a diagonal matrix with the enclosed vector as the set of diagonal elements. $z_1\sim \mathcal{C}\mathcal{N}(\mu_1,\sigma_1^2)$ and $\boldsymbol{z}_2\sim \mathcal{C}\mathcal{N}(\boldsymbol{\mu}_2,\boldsymbol{R}_{2})$ denotes that $z_1$ and $\boldsymbol{z}_2$ are scalar and vector versions of circular symmetric complex Gaussian (CSCG) random variables with mean $\mu_1$, variance $\sigma_1^2$ and mean vector $\boldsymbol{\mu}_2$, covariance matrix $\boldsymbol{R}_{2}$ respectively. $\mathbb{E}[\cdot]$ denotes the statistical expectation operator. $\mathbb{R}$  and $\mathbb{C}$ stands for the real and complex sets respectively.  $\boldsymbol{1}_{\{\cdot\}}$ stands for an indicator function that takes value 1 when enclosed condition stands true and 0 when not. $G=(V,E)$ denotes a graph with the set of vertices $V$ and set of edges $E$.

\section{System Model}
\label{sec:sysmodel}

We consider a multi-cell wireless network served by a BS located at $\boldsymbol{o}_b\in\mathbb{R}^2$ as illustrated in Fig.~\ref{fig:sysmodel}. To eliminate inter-cell interference and simplify the analysis, we assume an orthogonal frequency reuse scheme among adjacent cells. Hence, without loss of generality, we focus our investigation on a single representative cell. We denote the cell using $\mathcal{O} \subset \mathbb{R}^2$. Within this cell, the scenario is modeled as a downlink multi-user multiple-input single-output (MU-MISO) system, where the $M$-antenna BS serves $K$ single-antenna users. The communication is assisted by $J$ RISs. 
%The cells are assumed to be circular for simplicity. Furthermore,  assume that all cells operate using orthogonal frequencies, which makes it sufficient to investigate just one of these cells.  Figures/Results

\begin{figure}[t!]
\centerline{\includegraphics[width=0.34\textwidth]{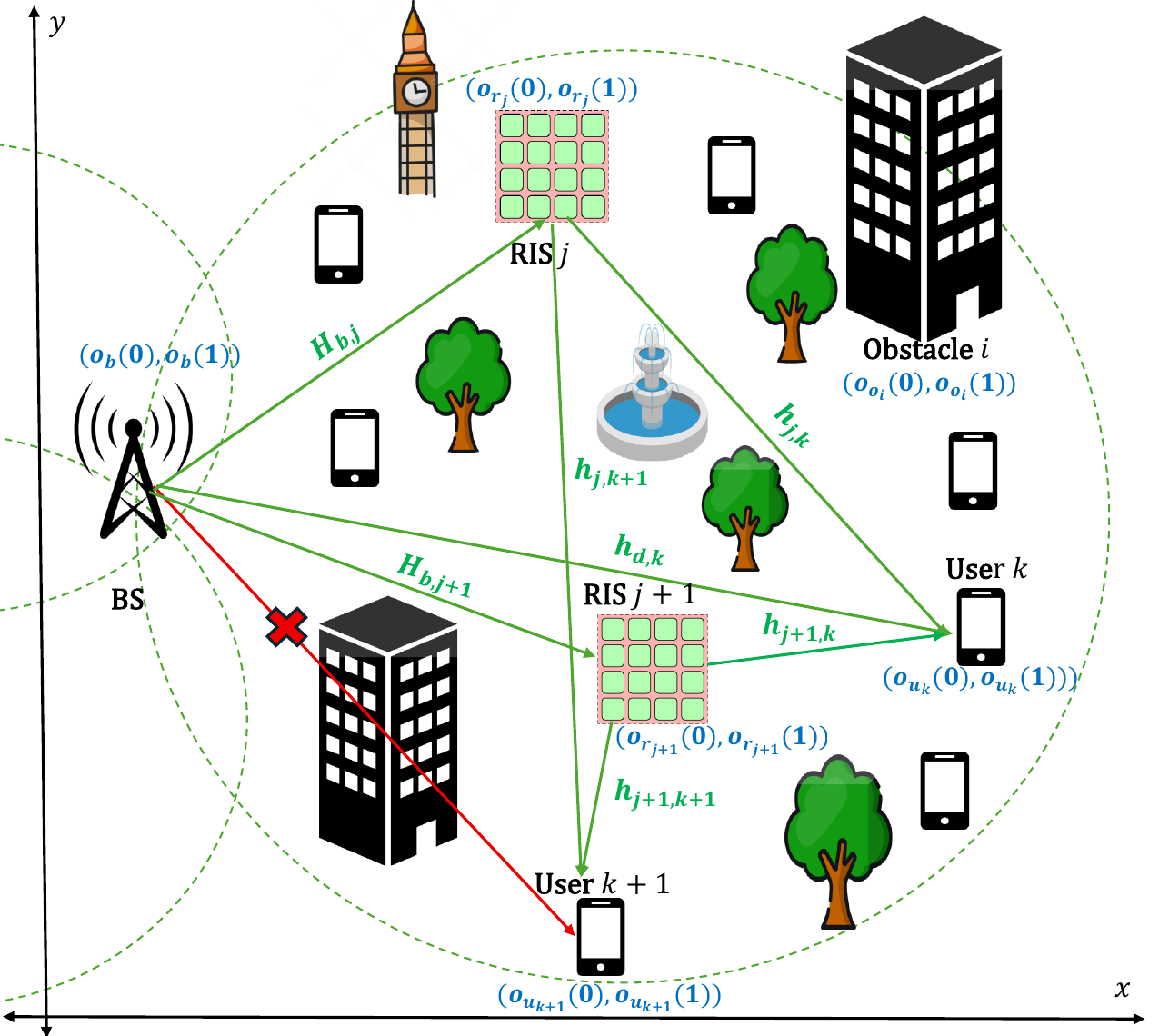}}
\caption{Multi-RIS-assisted MU-MISO system with obstacles. \label{fig:sysmodel}}
\end{figure} %Figures/

%Thus, looking at one such cell served by the BS, a multi-user multiple-input single-output (MU-MISO) system is considered where the BS serves multiple single-antenna users.
We also assume the presence of $I$ obstacles, some of which obstruct the line-of-sight (LoS) links between the BS and the users. For modeling purposes, all these obstacles are represented by two types: (1) circular such as pillars and (2) wall-type such as buildings \cite{loc1}. Let $\{\boldsymbol{o}_o\} \subset \mathcal{O}$ denote the set of centers of all obstacles. Circular-type obstacles are also associated with their radii $\{r_{ci}\}$. On the other hand, the wall-type obstacles are associated with their lengths $\{\ell_{wi}\}$ and orientations $\{\theta_{wi}\}$. We also assume that all obstacles are identical about the $z$-axis (making it sufficient to investigate just with their $x$ and $y$ coordinates). Thus, $(\{\boldsymbol{o}_o\},\{r_{ci}\},\{\ell_{wi}\}, \{\theta_{wi}\})$ is defined as the obstacle configuration. It is assumed to be \emph{static} and \emph{known} in the region\footnote{From the practical standpoint, this information is fully available to us once the region of deployment (like a specific school) is decided (then the positions of the buildings, stadiums, etc. in the school are known). Therefore, assuming knowledge of the obstacle configuration at the deployment stage is reasonable.}.

% Inhomogeneous Poisson Point Process (PPP) provides greater versatility in representing the spatial variability in dense environments~\cite{loc3}. We model the user locations by an inhomogeneous PPP $\Psi_u = \{\boldsymbol{o}^{(k)}_u\} \subset \mathcal{O}$ with intensity $\lambda_u(\boldsymbol{o})$ and total intensity $\Lambda_u =\int_{\mathcal{O}} \lambda_u(\boldsymbol{o})\,d\boldsymbol{o}$. Since their locations is a PPP, these user locations are i.i.d.\ with density
% \begin{equation}
%     f_u(\boldsymbol{o}) = \frac{\lambda_u(\boldsymbol{o})}{\Lambda_u}.
%     \label{eq:fu_def}
% \end{equation}

In practice, the exact locations of users are not available pre-deployment. Thus, we introduce an inhomogeneous spatial point process $\Psi_u = \{\boldsymbol{o}_u^{(1)}, \ldots, \boldsymbol{o}_u^{(K)}\}$ for modeling the user variability, of which several models (e.g., Poisson, Binomial, etc) are special cases. Here, $K$ denotes the number of users in the cell $\mathcal{O}$ drawn from some distribution $K \sim f_K$ with mean $\mathbb{E}[K] = \Lambda_u = \int_{\mathcal{O}} \lambda_u(\boldsymbol{o})\,d\boldsymbol{o}$, where $\lambda_u(\boldsymbol{o})$ denotes the spatial intensity at a point $\boldsymbol{o}\in \mathbb{R}^2$ over $\mathcal{O}$. For an instance of $K$, the user locations 
are assumed to be independently and identically distributed with the given probability density function (pdf)
\begin{equation}
   \boldsymbol{o}_u^{(1)}, \ldots, \boldsymbol{o}_u^{(K)} \sim  f_u(\boldsymbol{o}) = \frac{\lambda_u(\boldsymbol{o})}{\Lambda_u}, \boldsymbol{o} \in \mathcal{O}.
    \label{eq:fu_def}
\end{equation}
In words, the density at each point in the considered cell is the spatial intensity at the point normalized over the total spatial intensity of the cell. For instance, if this spatial intensity is set to be equal throughout the cell, it would be the uniform case. In practice, historic user data can be used to determine these intensity parameters and the distribution of $K$. 

%using historic data, one can deduce the spatial variability in user density across the cell. Unlike conventional spatial distributions like a 2D Gaussian or uniform models, we introduce a 

% This formulation is deliberately general: the proposed framework only 
% requires knowledge of the spatial density $f_u(\boldsymbol{o})$ in 
% \eqref{eq:fu_def} and is agnostic to the specific choice of $\mathcal{D}_K$ 
% governing the number of users. For instance, choosing $K \sim 
% \text{Poisson}(\Lambda_u)$ recovers the inhomogeneous Poisson point process 
% (PPP) model, under which the user locations $\Psi_u$ are i.i.d.\ with density 
% $f_u(\boldsymbol{o})$ as in \eqref{eq:fu_def}, and $K$ and $\Psi_u$ are 
% mutually independent~\cite{loc3}. Other choices, such as a fixed or binomially 
% distributed $K$, yield a binomial point process (BPP) with the same spatial 
% density $f_u(\boldsymbol{o})$. The specific instantiation used for numerical 
% evaluation is detailed in Sec.~\ref{sec:simulations}.

% \subsection{Channel Modeling}
% \label{sec:sysmodel}
For the upcoming discussion on the link-level model of the considered system, we consider a particular realization of the user process, $\Psi_u$. Let the users be indexed as $k \in \{1, \ldots, K\}$ where $K=|\Psi_u|$ and each $k$ corresponds to a location $\boldsymbol{o}^{(k)}_{u} \in \Psi_u$. The direct link from BS to user $k$ is Rayleigh modeled
\begin{equation}
   \boldsymbol{h}_{d,k} = \sqrt{\beta_{d,k}}\tilde{\boldsymbol{h}}_{d,k},
%\sqrt{\beta_0(d_{b{u}})^{-\kappa_{bu}}}
\label{eq:1}
\end{equation}
where $\beta_{d,k}$ denotes linear path loss and $\tilde{\boldsymbol{h}}_{d,k} \in \mathbb{C}^{M\times 1} \sim \mathcal{C}\mathcal{N}(\boldsymbol{0},\boldsymbol{I}_M)$. To illustrate the effect of obstacles in a next-generation obstacle-deterred environment, we assume that when the link is obstructed by an obstacle, the link power, $\sqrt{\beta_{d,k}}$ is reduced to a negligible level due to penetration loss.

We assume each RIS has $N$ passive reflecting elements located at $\boldsymbol{O}_{r} = \left[\boldsymbol{o}^{(1)}_{r},\ldots,\boldsymbol{o}^{(J)}_{r}\right] \in \mathbb{R}^{2\times J}$ such that $\boldsymbol{o}^{(j)}_{r}\in\mathcal{O}$ for each $j \in \{1,\ldots,J\}$. The links between BS and each RIS~$j$ and between each RIS~$j$ and the $k^{th}$ user are Rician
\begin{align}
    \boldsymbol{H}_{b,j} &= \sqrt{\beta_{b,j}}\left(\sqrt{\frac{T_1}{1+T_1}}\boldsymbol{\bar{H}}_{b,j}+ \sqrt{\frac{1}{1+T_1}}\tilde{\boldsymbol{H}}_{b,j}\right),\\
    \boldsymbol{h}_{j,k} &= \sqrt{\beta_{j,k}}\left(\sqrt{\frac{T_2}{1+T_2}}\boldsymbol{\bar{h}}_{j,k}+ \sqrt{\frac{1}{1+T_2}}\boldsymbol{\tilde{h}}_{j,k}\right) ,
\end{align}
where $\beta_{b,j}, \beta_{j,k}$ denote the path loss components. $\tilde{\boldsymbol{H}}_{b,j} \in \mathbb{C}^{N\times M}$, $ \boldsymbol{\tilde{h}}_{j,k}\in \mathbb{C}^{N\times1}$ and $\boldsymbol{\bar{H}}_{b,j}=\boldsymbol{a}_N(\vartheta
)\boldsymbol{a}_{M}^{\mathrm{H}}(\phi)$, $\boldsymbol{\bar{h}}_{j,k} = \boldsymbol{a}_N(\zeta_k)$ respectively denote the unit-variance CSCG and deterministic components of the links. Here, $\boldsymbol{a}_i$ is the steering vector of size $i$ for $i\in\{N,M\}$ and $\vartheta, \phi,\zeta_k$ are the angular parameters. $T_1$ and $T_2$ are the Rician parameters in linear scale. Like the direct link, path loss coefficients follow the same model when an obstacle lies on the respective link. Furthermore, the reflection matrix of RIS $j$ is
%%% CSCG defined in notations!
\begin{align}
    \boldsymbol{\Theta}_j =  \text{diag}(\boldsymbol{\theta}_j),
\end{align}
where $\boldsymbol{\theta}_j = [\theta^{(1)}_{j},\ldots,\theta^{(N)}_{j}]^T$, such that $\theta^{(n)}_{j} = e^{i\varphi^{(n)}_{j}}$ is the phase-shift introduced at the $n^{th}$ element of RIS $j$, for all $1 \le n \le N$. Let the combined signal transmitted from the BS be $\boldsymbol{x}=\sum_{i=1}^{K} \boldsymbol{w}_i x_i$ where $\boldsymbol{w}_i \in \mathbb{C}^{M \times 1}$ is the transmit beamformer and $x_i$ the unit-power transmit symbol corresponding to the $i^{th}$ user, for each $1\le i \le K$. We set the net indirect link as \( \boldsymbol{H}_{j,k} = \text{diag}(\boldsymbol{h}_{j,k}^{\mathrm{H}}) \boldsymbol{H}_{b,j} \). Thus, the net end-to-end link from BS and the corresponding received signal at $k^{\text{th}}$ user are respectively given by $\boldsymbol{h}_k = (\boldsymbol{h}_{d,k}^{\mathrm{H}}+\sum_{j=1}^J\boldsymbol{\theta}_j^T \boldsymbol{H}_{j,k})$ and 
\begin{equation}
y_k = \boldsymbol{h}_k \boldsymbol{x} + z_k.
\end{equation}
% \begin{align}
% y_k &= \left( \boldsymbol{h}_{d,k}^{\mathrm{H}} + \sum_{j=1}^{J}\boldsymbol{h}_{j,k}^{\mathrm{H}} \boldsymbol{\Theta}_j \boldsymbol{G} \right) \sum_{i=1}^{K} \boldsymbol{w}_i s_i + z_k,
% \end{align}
where \( z_k \sim \mathcal{CN}(0, \sigma^2) \) is the AWGN at the $k^{th}$ user. %%% N_0 is used in the channel estimation paper we have cited, hence for uniformity
Thus, the signal-to-interference-plus-noise ratio (SINR) at user~$k$ is
\begin{equation}
\gamma_k = \frac{\|\boldsymbol{h}_{k} \boldsymbol{w}_{k} \|^2}{\sum_{i \ne k} \| \boldsymbol{h}_{k} \boldsymbol{w}_{i} \|^2 + \sigma^2}.
\label{eq:SINR}
\end{equation}

\begin{remark}
Notably, in the aforementioned model, we have assumed that the links that involve reflections off more than one RIS simultaneously before landing at one of the users are low in power and thus negligible in modeling~\cite{loc1,loc2,loc3}. \label{rm:nohops}
\end{remark}
\section{Problem Formulation}
\label{sec:prblmform}
In this section, we formulate an optimization problem to obtain optimal multi-RIS placement while accounting for the throughput--fairness tradeoff. Firstly, we define the optimal instantaneous sum rate of the considered system for an RIS deployment $ \boldsymbol{O}_{r}$, a realization $\Omega = \left( \Psi_u, \{\boldsymbol{h}_k\}_{k=1}^K \right)$ of the user process and their corresponding random channels to be
\begin{subequations}
\begin{align}
\mathbf{P0}: R^*(\boldsymbol{O}_{r},\Omega) &= \max_{\boldsymbol{W}, \boldsymbol{\Theta}} \sum_{k=1}^{K} \log(1 + \gamma_k) \nonumber\\
\text{s.t.} \quad & \left|\theta_j^{(n)}\right| = 1, \quad \forall \ j, n, \label{subeq:RIS3}\\
& \sum_{k=1}^{K} \|\boldsymbol{w}_k\|_2^2 \le P_T. \label{subeq:power3}
\end{align}
\end{subequations}
Here, considering an instantiation of $\{\boldsymbol{h}_k\}_{k=1}^K$ at each instant, in wireless terms translates to quasi-static flat fading model where it remains static within a coherence time, $\tau_{tot}$. Similarly, considering an instantiation of the user process, $\Psi_u$ translates to the assumption that at a given time instant during the real-time optimization, user locations are deterministic and can be estimated while estimating CSI. $\boldsymbol{\Theta} = [\boldsymbol{\theta}_1,\ldots,\boldsymbol{\theta}_J]\in\mathbb{C}^{N\times J}$ is the accumulated phase matrix of the RISs and $\boldsymbol{W} = [\boldsymbol{w}_1,\ldots,\boldsymbol{w}_K]\in \mathbb{C}^{M\times K}$ the accumulated beamformer matrix. The feasibility set for the generation of phases at each of the $N$ elements of each of the $J$ passive RISs is given in~\eqref{subeq:RIS3} and ~\eqref{subeq:power3} presents the total power constraint for the beamformer design where $P_T$ is the maximum total transmit power. Then, the problem of interest is formulated as given below
\begin{subequations}
\begin{align}
\mathbf{P1}: J^* = & \min_{\boldsymbol{O}_{r}} \quad  J(\boldsymbol{O}_{r}) \nonumber\\
\text{s.t.} \quad & p_{\mathrm{cov}}( \boldsymbol{O}_{r}) \ge p_{c_{th}}, \label{subeq:coverage}\\
& L\left(\boldsymbol{o}_b, \boldsymbol{o}_{\mathrm{r}}^{(j)};\boldsymbol{O}_{r}\right) = 1, \quad \forall  \ j. \label{subeq:BS}
\end{align}
\end{subequations}
\begin{subequations}
\begin{align}
\mathbf{P2}: \boldsymbol{O}_{r}^* = &\arg\max_{\boldsymbol{O}_{r}}\mathbb{E}_{\Omega}\left[  R^*( \boldsymbol{O}_{r}, \Omega) \right] \nonumber\\
\text{s.t.} \quad & J(\boldsymbol{O}_{r}) = J^*\label{subeq:J},\\
&p_{\mathrm{cov}}(\boldsymbol{O}_{r}) \ge p_{c_{th}}, \label{subeq:coverage2}\\
 & L\left(\boldsymbol{o}_b, \boldsymbol{o}_{r}^{(j)};\boldsymbol{O}_{r}\right) = 1, \quad \forall \ j. \label{subeq:BS2}
\end{align}
\end{subequations}
The first subproblem, $\mathbf{P1}$ is formulated to compute the least number of RISs  that can achieve the required minimum level of coverage from the BS. More precisely, we minimize $J=\mathrm{cols}(\boldsymbol{O}_r)$, i.e., the number of columns of $\boldsymbol{O}_r$ for varying  RIS placements $\boldsymbol{O}_r$ that satisfy the minimum coverage requirements. We define coverage probability $p_{\mathrm{cov}}(\boldsymbol{O}_{r})$ to be the expected fraction of users that can be served in $\mathcal{O}$. This is mathematically defined using the user point process as
\begin{equation}
    p_{\mathrm{cov}}(\boldsymbol{O}_{r}) = \int_{\mathcal{O}} \boldsymbol{1}_{\left\{L_{\mathrm{net}}(\boldsymbol{o}; \boldsymbol{O}_{r}) \ge 1\right\}} f_u(\boldsymbol{o}) d\boldsymbol{o},\label{eq:covdef}
\end{equation}
where the probability is evaluated by integrating the pdf $f_u$ over the regions in the considered cell that are reachable from the BS by at least one path (directly or via one of RISs). This reachability condition is mathematically modeled using $L_{\mathrm{net}}(\boldsymbol{o};\boldsymbol{O}_{r}) =  L(\boldsymbol{o}_b, \boldsymbol{o})+\sum_{j=1}^J L(\boldsymbol{o}_b, \boldsymbol{o}_{r}^{(j)};\boldsymbol{O}_{r})L(\boldsymbol{o}_{r}^{(j)}, \boldsymbol{o};\boldsymbol{O}_{r})$ which is the number of paths that can reach a point $\boldsymbol{o}$ from the BS, for a given RIS deployment $\boldsymbol{O}_{r}$, where $L(\boldsymbol{o}_1,\boldsymbol{o}_2)$ is an indicator that takes $1$ and $0$ when an LoS path exists and does not exist respectively between any $\boldsymbol{o}_1,\boldsymbol{o}_2 \in \mathbb{R}^2$. $p_{c_{th}}$ is set as the minimum probability of coverage that we wish to guarantee in~\eqref{subeq:coverage}. Equation~\eqref{subeq:BS} further clarifies the search space by necessitating an LoS link from the BS to each point where an RIS could be potentially placed. We denote these two constraints together as the dual-visibility requirement. 

Having obtained the least possible $J$ and the corresponding search space where RISs need to be placed to guarantee the required coverage in $\mathbf{P1}$, we formulate $\mathbf{P2}$ to compute that RIS configuration that is most probable to achieve the highest throughput while satisfying the conditions obtained from $\mathbf{P1}$. Formally, we propose to maximize the expectation of optimal sum rate over the random channels and the user process while satisfying the coverage conditions from $\mathbf{P1}$. %$\tau_{ov}$ is the time for channel estimation which reduces the achievable rate. Since we assume orthogonal feedback resources, this factor is not a function of $K$ because we are able to send and receive symbols at the same time from/to all users from BS.

The coverage-throughput tradeoff can be controlled by varying $p_{c_{th}}$ in the proposed formulation $\mathbf{P0}$--$\mathbf{P2}$. For instance, if set close to 100 \%, then problem attaches maximum importance to coverage, resulting in a much reduced space that satisfies the coverage constraints, reducing the impact that throughput maximization can have because it is constrained within this reduced search space. On the contrary, with low $p_{c_{th}}$, the proposed problem provides a larger search space for throughput maximization, increasing its impact.
%We formulate $\mathbf{P1}$ to compute the minimum number of RISs needed and the corresponding search space where we can deploy RISs to guarantee $p_{c_{th}}$. This conveniently helps us overcome the integer problem and multi-objective problem if $\mathbf{P1}$ and $\mathbf{P2}$ are combined as a single problem.
%modularity of being able to employ existing or future beamforming techniques is a key feature of the proposed problem. 
%Thus, the proposed formulation maximizes expected system throughput while minimizing $J$ subject to a dual-visibility requirement. It is noteworthy that the proposed formulation $\mathbf{P2}$ is general that can be used to maximize both sum rate performance and coverage based on the tradeoff parameter, $p_{c_{th}}$. If set close to 100\%, then problem mimics a simple coverage maximization problem, while if we set close to $0$\%, it turns into a purely throughput based placement problem.

\begin{remark}
This formulation is a `softly-coupled' optimization unlike a joint one. In Section \ref{sec:intro}, we argued that the joint case isn't practical for an RIS-assisted system since we need $\boldsymbol{O}_r$ before deployment and that it cannot be allowed to be optimized in real-time with other variables that need to be optimized in real-time. Also, in the purely joint case, every variable of optimization is re-optimized based on the optimal values of every other, alternatively. But in our formulation, in $\mathbf{P2}$, while the optimal RIS placement $\boldsymbol{O}_{r}$ gets re-optimized based on $J$ but $J$ is not re-optimized based on $\boldsymbol{O}_{r}$. As a result, though $\boldsymbol{O}_{r}$ is used in determining $J$ while assuring coverage, we allow it to be re-optimized with respect to throughput, allowing the RIS positions to align itself in order to provide coverage to overcome the obstacles and also exploit the spatial user density to maximize throughput. $\boldsymbol{\Theta}$ and $\boldsymbol{W}$ are set to independently maximize real-time throughput without any coverage constraints and $J$ is set to independently optimize coverage. We also ensure that the formulation avoids a multi-objective problem with an integer optimization in $J$. 
\end{remark}

% \begin{equation}
% \begin{aligned}
% \mathbf{P2}:&\max_{J, \boldsymbol{O}_{\mathrm{r}}, \boldsymbol{\Theta}, \boldsymbol{W}} \quad  \sum_{k=1}^{K}\mathrm{P}\left\{ L^{(k)}_{\mathrm{net}}\ge1 \right\}  \\
% \text{s.t.} \quad & J = \min_{J,\boldsymbol{O}_r,\boldsymbol{\Theta},\boldsymbol{W}} J,\\
% & L\left(\boldsymbol{o}_b, \boldsymbol{O}_{\mathrm{r}}^{(g)}\right) = 1, \quad \forall \ g, \\
% & \left|\theta_j^{(n)}\right| = 1, \quad \forall \ j, n, \\
% & \sum_{k=1}^{K} \|\boldsymbol{w}_k\|^2 \le P_T.
% \end{aligned}
% \end{equation}

% \begin{equation}
% \begin{aligned}
% &\max_{\boldsymbol{O}_{\mathrm{r}}, \boldsymbol{\Theta}, \boldsymbol{W}} \mathbb{E}\left[\sum_{k=1}^{K} \log ( 1 + \gamma_k )\right] \\
% \text{s.t.} \quad 
% & \left|\theta_g^{(n)}\right| = 1, \quad \forall \ g, n, \\
% & \sum_{k=1}^{K} \|\boldsymbol{w}_k\|^2 \le P_T.
% \end{aligned}
% \end{equation}

% Here $L(o_1, o_2) = 1$ stands for the ability of a link to be established between $o_1$ and $o_2$. Also notice that the lower bound on probability of existence of a link between an RIS and any user has been fixed to be 95\%.

\section{Search Space Determination and Resource Minimization}
\label{sec:ss}
Firstly, a technique is proposed to compute the space spanned by constraints~\eqref{subeq:coverage} and~\eqref{subeq:BS}. Then, an elegant method is proposed to solve the problem, $\mathbf{P1}$. %the resource minimization problem

%Finally, a technique is proposed for solving the performance maximization problem $\mathbf{P1.2}$ with the optimal number of RISs obtained previously. 
\subsection{Preliminary Processing}
\subsubsection{Discretization} To optimize the placement of the RISs, we first discretize the continuous cell $\mathcal{O}$ into a finite and uniform grid of points $\mathcal{S} = \{\boldsymbol{s}_1, \boldsymbol{s}_2, \dots, \boldsymbol{s}_{N_{\text{grid}}}\}\subset \mathbb{R}^2$ where each $\boldsymbol{s}_i$ corresponds to a spatial region $A_i$ such that the corresponding intensity is $\Lambda(\boldsymbol{s}_i) =\int_{A_i}\lambda_u(\boldsymbol{o})d\boldsymbol{o}$ and the corresponding number of users, $N(A_i) \sim f_K(\Lambda(\boldsymbol{s}_i))$. A subset of these points correspond to the obstacle configuration, $\mathcal{Q} \subset \mathcal{S}$. We also define $\boldsymbol{s}_b \in \mathcal{S}$ to be the discretized version of $\boldsymbol{o}_b$ defined earlier. We define `visible set' as the set of vertices that have a direct LoS to the BS, $\mathcal{V} = \{\boldsymbol{s}\in \mathcal{S}\;|\; L(\boldsymbol{s}_b,\boldsymbol{s}) = 1\}$.%\setminus\mathcal{Q}

Finally, we define $\mathcal{S}^* = \{s^*\in \mathcal{S} | \Lambda(s^*)> 0 \} \subset \mathcal{S}$ to be a subset of $\mathcal{S}$ which have a positive user intensity as per the defined user distribution $f_u$. This is the set of points we need to extend coverage from the BS. We also define $\mathcal{S}^\prime=\mathcal{S}\setminus\mathcal{S}^*$ which are those points which have non-positive user intensity.
\subsubsection{Obstacle Clusters} Instead of handling $I$ obstacles and the areas they shadow individually, we form $T \ll I$ physically meaningful clusters among the obstacle points in $\mathcal{Q}$ using the Density-Based Spatial Clustering of Applications with Noise (DBSCAN) algorithm \cite{dbscan} which is particularly effective here as it does not require the number of clusters to be pre-defined and is effective with a wide variety of shapes. These clusters $\mathcal{Q}_i\text{ }\forall 1 \le i \le T$ are formed based on two DBSCAN parameters: neighborhood radius, $\epsilon$ and a lower bound on the number of points in a cluster $N_{\mathrm{DBSCAN}}$. This allows us to treat each of these clusters as ``large obstacles'', considerably reducing computations as shown in Section \ref{subsec:complexity}. Assuming all points are clustered, these form a disjoint partition $\mathcal{Q}_i \cap \mathcal{Q}_j = \emptyset$  for all $i \neq j$, $1 \le i,j \le T$, and $\bigcup_{i=1}^T\mathcal{Q}_i = \mathcal{Q}$. %Default numbers for these DBSCAN parameters are given while discussing numerical results,which is the typical case,
\subsection{Search Space Determination} 
\label{subsubsec:const}
An undirected visibility graph \cite{graphtheory} $G_{\mathrm{vis}} = (V_{\mathrm{vis}}, E_{\mathrm{vis}})$ with $V_{\mathrm{vis}} = \mathcal{S}$ and $E_{\mathrm{vis}} = \{(\boldsymbol{s}_1,\boldsymbol{s}_2),(\boldsymbol{s}_2,\boldsymbol{s}_1) \;|\; L(\boldsymbol{s}_1,\boldsymbol{s}_2) = 1,\boldsymbol{s}_1,\boldsymbol{s}_2\in V_{\mathrm{vis}}\}$ is considered. An edge exists between two vertices if no obstacle lies in the LoS path between them. 

Let the set $\mathcal{U} = \mathcal{S}^* \setminus \mathcal{Q} \setminus \mathcal{V}$ indicate the set of points where user intensity is positive (user can exist according to the spatial process assumed), an obstacle is absent and is not yet reachable from the BS. In other words, this is the set of points to which we need to maximize extension of coverage. Adopting a targeted mitigation strategy, we partition this set by cause of blockage. In our model, the cause of blockage for any point $\boldsymbol{u} \in \mathcal{U}$ is the presence of an obstacle, i.e., at least one point from one or more of the $T$ obstacle clusters $\{\mathcal{Q}_i\}_{i=1}^T$ lie on its link to the BS. For each $\boldsymbol{u}\in\mathcal{U}$, let
\begin{equation}
    \mathcal{Q}(\boldsymbol{u}) = \{\boldsymbol{q}_1,\dots,\boldsymbol{q}_{N_{\boldsymbol{u}}}\} \subseteq \mathcal{Q}, \qquad N_{\boldsymbol{u}} \ge 1,
\end{equation}
be the ordered set of obstacle points on the LoS link between BS and $\boldsymbol{u}$ as indicated below
\begin{equation}
\{(\boldsymbol{s}_b,\boldsymbol{q}_1), (\boldsymbol{q}_1,\boldsymbol{q}_2),\dots,(\boldsymbol{q}_{N_{\boldsymbol{u}}}, \boldsymbol{u})\} \subset E_{\mathrm{vis}}.
\end{equation}
$N_{\boldsymbol{u}} \ge1$ for any $\boldsymbol{u}\in\mathcal{U}$ because every point in the set is not yet visible from the BS, by definition. We then define the set of points in $\mathcal{U}$ that are blocked `first' by obstacle $\mathcal{Q}_i$ as
\begin{equation}
    \mathcal{U}_i = \big\{ \boldsymbol{u} \in \mathcal{U} \mid \boldsymbol{q}_1 \in \mathcal{Q}_i \big\}, \qquad 1 \le i \le T.
\end{equation}
Since $\mathcal{Q} = \bigcup_{i=1}^T \mathcal{Q}_i$ and $\mathcal{Q}(\boldsymbol{u}) \neq \emptyset \text{ } \forall\boldsymbol{u}\in\mathcal{U}$, $\bigcup_{i=1}^T \mathcal{U}_i = \mathcal{U}$. Importantly, it is a disjoint union since uniquely characterized by one `first' blocking point $\boldsymbol{q}_1$ belonging to exactly one $\mathcal{Q}_i$. 

Since we have $T$ sets of points/discretized regions $\{\mathcal{U}_i\}_{i=1}^T$ shadowed by $T$ ``large obstacles'', it is easy to show that we need a maximum of $T$ RISs to extend coverage to each of these sets \cite{loc3}. Firstly, we assume we need all $T$ RISs. We need to compute $T$ preliminary candidate sets, each of which denote the set of all possible locations for the corresponding RIS to be able to extend coverage to a `considerable percentage' of the respective shadowed region. Equation~\eqref{subeq:coverage} defines the constraint for deciding the percentage. We state and prove the following theorem to achieve the coverage constraint.

\begin{theorem}
\label{lm:cov}
Assuring $p_{c_{th}}$ percent coverage in each shadowed region $\mathcal{U}_i, 1 \le i \le T$ as given by
\begin{equation}
\frac{\sum_{\tilde{\boldsymbol{s}}\in\mathcal{U}_i} \boldsymbol{1}_{\{L_{\mathrm{net}(\tilde{\boldsymbol{s}})} \geq 1\}} \Lambda(\tilde{\boldsymbol{s}})}{\sum_{\tilde{\boldsymbol{s}}\in\mathcal{U}_i}\Lambda(\tilde{\boldsymbol{s}})}  \geq p_{c_{\mathrm{th}}} 
    \label{eq:local}
\end{equation}
guarantees the global constraint~\eqref{subeq:coverage}.
\end{theorem}

\begin{IEEEproof}
    Refer Appendix \ref{app:lm1}.
\end{IEEEproof}

It is notable that this reduces the global requirement to a more conservative local requirement. Therefore the achieved coverage results, as we'll see later, would be better than what is set as the required minimum. Then, we have a proposition to compute the preliminary candidate sets using Theorem \ref{lm:cov}.
\begin{proposition}
\label{prop:1}
Each Preliminary candidate set $\mathcal{C}^{\prime}_i, 1\le i\le T$ that correspondingly assures $p_{c_{th}}$ percent coverage of $\mathcal{U}_i$ can be computed over a subgraph of $G_{\mathrm{vis}}$ with vertices $\mathcal{U}_i$ using 
%&= \left\{ \boldsymbol{s} \in V_{\mathrm{vis}} \;\middle|\; (\boldsymbol{s}, \boldsymbol{s}_{b}) \in E_{\mathrm{vis}},\frac{\sum_{\boldsymbol{u}_i \in \mathcal{U}_i} \boldsymbol{1}_{\{(\boldsymbol{s}, \boldsymbol{u}_i)\in E_{\mathrm{vis}}\}}f_u(\boldsymbol{u}_i)}{\sum_{\boldsymbol{u}_i\in\mathcal{U}_i}f_u(\boldsymbol{u}_i)} \geq p_{c_{\mathrm{th}}} \right\},\\
\begin{align}
\mathcal{C}^{\prime}_i = \left\{ \boldsymbol{s} \in \mathcal{V} \;\middle|\; \frac{\sum_{\boldsymbol{u} \in \mathcal{U}_i} \boldsymbol{1}_{\{(\boldsymbol{s}, \boldsymbol{u})\in E_{\mathrm{vis}}}\}\Lambda(\boldsymbol{u})}{\sum_{\boldsymbol{u}\in\mathcal{U}_i}\Lambda(\boldsymbol{u})} \geq p_{c_{\mathrm{th}}} \right\}.
\label{eq:candidate_set}
\end{align}
By Theorem \ref{lm:cov}, we have that when $T$ RISs are placed in a location from each of these sets, they collectively guarantee $p_{c_{th}}$ percent coverage for the considered wireless system.
% \begin{equation}
% \mathcal{C}'i = \left\{ \boldsymbol{s} \in \mathcal{V} \;\middle|\; \frac{\sum{\boldsymbol{u} \in \mathcal{U}i} 1{{(\boldsymbol{s},\boldsymbol{u}) \in E_{\mathrm{vis}}}} f_u(\boldsymbol{u})}{\sum_{\boldsymbol{u} \in \mathcal{U}i} f_u(\boldsymbol{u})} \geq p{c_{\mathrm{th}}} \right\}.
% \end{equation}
% where $\boldsymbol{1}_{\{.\}}$ represents the indicator function which takes value $1$ when the enclosed condition is true and $0$ if not.
\end{proposition}
\begin{IEEEproof}
Refer Appendix \ref{app:prop1}.
\end{IEEEproof}
This set can be manually computed by searching over the set $\mathcal{V}$ by checking for the local coverage requirement. 

\subsection{Minimizing the number of RISs} 
Although when RISs are placed in locations from these preliminary candidate sets, we assure coverage to each shadowed $\mathcal{U}_i$, most of these RISs not only extend coverage to the respective $\mathcal{U}_i$ but also to other shadowed regions $\mathcal{U}_j,~i\ne j$. In this section, we aim to reduce this spatial redundancy in the collection of candidate sets to the lowest possible level, while achieving the same coverage. In other words, having obtained a collection of candidate sets where we can deploy $T$ RISs to guarantee the coverage requirements in the constraints of $\mathbf{P1}$, we now solve $\mathbf{P1}$ while lying within the constraint space. %in each shadowed region. 

To do that, we aim to compute a new, smaller collection of candidate sets ($J<T$ sets) where if an RIS is placed in each of these new sets, they together guarantee same $p_{c_{th}}$ percent coverage in all $T$ shadowed regions $\{\mathcal{U}_i\}_{i=1}^T$ and hence by Theorem \ref{lm:cov} guarantee the global coverage requirement~\ref{subeq:coverage}. 

A hypergraph is a generalization of a graph consisting of hyperedges which can exist between more than two vertices~\cite{hypergraph}. A graph with vertices as sets and edges between them indicating the existence of a non-empty intersection is called an intersection graph~\cite{graphtheory}. We define an intersection hypergraph $G_{\mathrm{int}} = (V_{\mathrm{int}},E_{\mathrm{int}})$ in which each of the $T$ vertices are the preliminary candidate sets,  $V_{\mathrm{int}} = \{C^\prime_1, C^\prime_2, \dots, C^\prime_T\}$ and a hyperedge $(C^\prime_{j_1}, \dots, C^\prime_{j_f})\in E_{\mathrm{int}}$ if $\bigcap_{i=j_1}^{j_f}C^{\prime}_i\ne \emptyset$ for $f\ge1$. We assume $C^{\prime}_i \ne \emptyset, \forall i$. We then have a proposition to find a solution to $\mathbf{P1}$ using the framework constructed.

%after applying the conservative local coverage constraint in Lemma \ref{lm:cov}/Theorem \ref{prop:1},
\begin{proposition}
\label{prop:2}
We propose to partition the vertex set $V_{\mathrm{int}}$ into subsets such that each of these subsets are hyperedges in $G_{\mathrm{int}}$ and the number of subsets is minimized. Thus, the collection of preliminary candidate sets is partitioned such that each subset of the collection has a non-empty intersection and such that the number of partitions is minimized. Let the partition be $\mathcal{P} = \{e_1,\ldots,e_J\}$ such that $e_i\subset V_{\mathrm{int}}, 1\le i \le J$ and $\bigcup_{i=1}^Je_i=V_{\mathrm{int}}$. The goal is to minimize $J=|\mathcal{P}|$ %while ensuring each subset $e_i, 1\le i \le J$ corresponds to a hyperedge in the graph, i.e., has a non-zero intersection between the vertices in the subset. This problem is given by
\begin{align}
    &\min_{\mathcal{P}} |\mathcal{P}|\\
    \text{s.t.}\quad & e_i \in E_{\mathrm{int}}, \forall e_i \in \mathcal{P}.
\end{align}
The corresponding final candidate sets $C_i, 1 \le i \le J$ are
\begin{equation}
    C_i = \bigcap_{j\in e_i}C^{\prime}_j
\end{equation}
\end{proposition}

\begin{IEEEproof}
Refer Appendix \ref{app:prop2}.
\end{IEEEproof}

Taking inspiration from the Greedy Maximal Scheduling (GMS) algorithm from network theory~\cite{GMS}, we propose Algorithm \ref{algo:greedy} to estimate this partition. We pick the hyperedge that has the most number of vertices i.e., non-empty intersection between largest number of preliminary candidate sets. We compute this intersection and set it as the first final candidate set for RIS placement. When we place an RIS in this region, we simultaneously assure $p_{c_{th}}$ percent coverage to all shadowed regions corresponding to the preliminary candidate sets in the selected hyperedge. We remove all vertices in this hyperedge from $G_{\mathrm{int}}$ and repeat the aforementioned steps until we have a graph with no vertices. Notably, while defining, we allowed singleton hyperedges in $G_{\mathrm{int}}$. Thus, if a $C^\prime_i$ has no non-empty intersection with any other set, then we simply add this set to the final collection of sets. Finally, we end up with a smaller number of final candidate sets that achieve $p_{c_{th}}$ percent coverage to $\{\mathcal{U}_i\}_{i=1}^T$. We denote the resulting number of final candidate sets (i.e., RISs) by $J_{\mathrm{greedy}}$.  Recall from Theorem \ref{lm:cov} that we adopted a conservative local constraint to compute $G_{\mathrm{int}}$, i.e., $J^*_{\mathrm{local}} \ge J^*$ where $J^*_{\mathrm{local}}$ denotes the optimal number of RISs needed to satisfy the modified constraint~\eqref{eq:local}. Also, since the described problem is a greedy heuristic to approach this $J^*_{\mathrm{local}}$, we have $J_{\mathrm{greedy}}\ge J^*_{\mathrm{local}}$. Thus, we get $J_{\mathrm{greedy}}\ge J^*_{\mathrm{local}}\ge J^*$. In summary, we propose a method similar to greedy set cover, greedy coloring, GMS, and use $J_{\mathrm{greedy}}$ as an approximation of $J^*$ in $\mathbf{P1}$.
\begin{algorithm}[t!]
\caption{Greedy maximal intersection partitioning}\label{algo:greedy}
\textbf{Input:}
\begin{enumerate}
    \item[(i)] The intersection hypergraph, $G_{\mathrm{int}}$,
\end{enumerate}
\textbf{Output:} 
\begin{enumerate}
    \item[(i)] Number of RISs $J_{\mathrm{greedy}}$,
    \item[(ii)] Corresponding candidate sets: $\{\mathcal{C}_i\}$, $1\le i \le J$
\end{enumerate} 
\begin{algorithmic}[1]
    \State Set $i = 0$
    \While{$G_{\mathrm{int}}$ has at least one vertex} 
        \State Set $i = i + 1$
         
        \State Pick hyperedge $e_{\mathrm{max}}$ in $G_{\mathrm{int}}$ with most vertices
        \State Compute the candidate set $C_i = \bigcap_{j\in e_{\mathrm{max}}} C^{\prime}_j$
        \State Remove the vertices $j\in e_{\mathrm{max}}$ from $G_{\mathrm{int}}$
        % \Else \State Set the final candidate set $C_i =  C^{\prime}_i$
        %\EndIf
    \EndWhile
    \State Set $J_{\mathrm{greedy}} = i$
\end{algorithmic}
\end{algorithm}
\section{Expected Sum Rate Maximization and the Overall Algorithm}
\label{sec:perf}

In this section, firstly, we discuss a joint beamforming technique to maximize system performance in $\mathbf{P0}$. Then, we discuss $\mathbf{P2} $ to finalize optimal RIS placement.

\subsection{Imperfect CSI, Channel Estimation, and RIS Training}
\label{subsec:imperfectcsi}

To implement the beamforming in real-time, CSI must be acquired. We propose to use a low-complexity estimation and RIS training scheme for the multi-RIS network.

\subsubsection{Direct Link Estimation}
With RISs in ``zero state'', let the BS transmit $M$ orthogonal pilots, $\boldsymbol{X}_p = \mathbf{I}_M$. User $k$ receives
\begin{equation}
\boldsymbol{y}_{d,k} = \sqrt{\frac{P_T}{M}}\boldsymbol{h}_{d,k}^{\mathrm{H}} + \boldsymbol{z}_{d,k},
\end{equation}
where $\boldsymbol{z}_{d,k} \sim \mathcal{C}\mathcal{N}(\mathbf{0},\sigma^2\mathbf{I}_M)$. We assume an analog feedback~\cite{channelestigen} of a scaled version of the observed symbol
\begin{align}
%  \hat{\boldsymbol{h}}_{{d,k}
% _{\text{DL}}} &= \frac{\sqrt{\frac{P_T}{M}}}{\frac{P_T}{M} + \sigma^2}\boldsymbol{y}^{\mathrm{H}}_{d,k},\\
\tilde{\boldsymbol{y}}_{d,k} &= \frac{\sqrt{\beta_{\mathrm{fb}}\frac{P_T}{M}}}{\sqrt{\beta_{d,k}\frac{P_T}{M}+\sigma^2}} \boldsymbol{y}_{d,k} +\tilde{\boldsymbol{z}}_{d,k}\\ 
&= \frac{\sqrt{\beta_{\mathrm{fb}}}\frac{P_T}{M}}{\sqrt{\beta_{d,k}\frac{P_T}{M}+\sigma^2}} \boldsymbol{h}^{\mathrm{H}}_{d,k} + \left(\frac{\sqrt{\beta_{\mathrm{fb}}\frac{P_T}{M}}}{\sqrt{\beta_{d,k}\frac{P_T}{M}+\sigma^2}}\boldsymbol{z}_{d,k}+\tilde{\boldsymbol{z}}_{d,k}\right)\\
\label{eq:directanalog}
&= \frac{\sqrt{\beta_{\mathrm{fb}}}\frac{P_T}{M}}{\sqrt{\beta_{d,k}\frac{P_T}{M}+\sigma^2}} \boldsymbol{h}^{\mathrm{H}}_{d,k} +\bar{\boldsymbol{z}}_{d,k},
\end{align}
where $\beta_{\mathrm{fb}}$ is the feedback gain and $\tilde{\boldsymbol{z}}_{d,k} \sim \mathcal{C}\mathcal{N}(\mathbf{0},\sigma^2\mathbf{I}_M)$ the feedback noise. For later convenience, let $A_k = \left(\sqrt{\beta_{\mathrm{fb}}}\frac{P_T}{M}\right)/\left(\sqrt{\beta_{d,k}\frac{P_T}{M}+\sigma^2}\right)$. Since $\boldsymbol{z}_{d,k}$ and $\tilde{\boldsymbol{z}}_{d,k}$ are both Gaussian, $\bar{\boldsymbol{z}}_{d,k}$ is too where each 0has variance
\begin{equation}
    \sigma_{\bar{Z},k}^2 = \sigma^2 \left( 1 + \frac{\beta_{\mathrm{fb}}\frac{P_T}{M}}{\beta_{d,k}\frac{P_T}{M}+\sigma^2} \right).
\end{equation}
Thus, the uplink MMSE estimate is\cite{channelestigen}% of the $k^{\mathrm{th}}$direct link is given by 
% Following standard minimum mean square error (MMSE) estimation and analog feedback over the uplink \cite{channelestigen}, the BS obtains the MMSE estimate,_{\text{UL}}
\begin{align}
\hat{\boldsymbol{h}}_{d,k} &= \frac{\sqrt{\beta_{\mathrm{fb}}}\frac{P_T}{M}\beta_{d,k}}{\left(\sqrt{\frac{P_T}{M}\beta_{d,k}+\sigma^2}\right)\left(\beta_{\mathrm{fb}}\frac{P_T}{M}+\sigma^2\right)}\tilde{\boldsymbol{y}}^{\mathrm{H}}_{d,k}.
    \label{eq:directest}
\end{align} %^{\mathrm{H}}
The estimation error $\boldsymbol{e}_{d,k} = \boldsymbol{h}_{{d,k}}- \hat{\boldsymbol{h}}_{d,k}$ can be shown to be Gaussian and independent of the estimate with the variance
\begin{equation}%\frac{1}{1 + \beta_{\mathrm{fb}}\frac{P_T}{M\sigma^2}} +
    \sigma_{e,k}^2 =  \beta_{d,k}\frac{1+\beta_{d,k}\frac{P_T}{M\sigma^2}+\beta_{\mathrm{fb}}\frac{P_T}{M\sigma^2}}{\left(1 + \beta_{\mathrm{fb}}\frac{P_T}{M\sigma^2}\right)\left(1 + \beta_{d,k}\frac{P_T}{M\sigma^2}\right)}.
\end{equation}
We assume FDD with orthogonal feedback resources available, enabling simultaneous transmission and reception for all users.

%where $\tilde{\boldsymbol{y}}_{d,k}$ is the feedback signal corrupted by AWGN $\tilde{\boldsymbol{z}}_{d,k}$. 

\subsubsection{Indirect Link Estimation} Let $\beta_{\mathrm{in}}=\beta_{b,j}\beta_{j,k}$ be the net path loss. The expected value of $\boldsymbol{H}_{j,k}$ is
\begin{equation}
\bar{\boldsymbol{H}}_{j,k}=\mathbb{E}[\boldsymbol{H}_{j,k}]=\sqrt{\frac{T_1T_2}{(1+T_1)(1+T_2)}}\sqrt{\beta_{\mathrm{in}}}\mathrm{diag}(\boldsymbol{\bar{h}}^{\mathrm{H}}_{j,k})\bar{\boldsymbol{H}}_{b,j}.    
\end{equation}
%which is fully determined by  path loss and angular parameters.
To estimate $\boldsymbol{H}_{j,k}$, RIS $j$ is activated while others are turned off. A unitary Discrete Fourier Transform (DFT) matrix $\boldsymbol{V} = [\boldsymbol{\theta}_{j,1}, \dots, \boldsymbol{\theta}_{j,N}]^T\in\mathbb{C}^{N\times N}$ with $\boldsymbol{\theta}_{j,1}=\boldsymbol{1}^T_N$ is used for training \cite{channelesti2,ristraining1,ristraining2}. In each of the $N$ iterations, $\boldsymbol{\theta}_{j}$ is set to a row of $\boldsymbol{V}$. The signal at user $k$ through RIS $j$ for $\boldsymbol{X}_p = \mathbf{I}_M$ is %is
\begin{align}
\boldsymbol{y}^{(n)}_{j,k} &= \sqrt{\frac{P_T}{M}}(\boldsymbol{h}_{d,k}^{\mathrm{H}} + \boldsymbol{\theta}^T_{j,n}\boldsymbol{H}_{j,k}) + \boldsymbol{z}^{(n)}_{j,k}, \quad 1\le n \le N.
\end{align}
Then, we adopt a similar analog feedback procedure like~\eqref{eq:directanalog}. Additionally, BS isolates the indirect link by subtracting a scaled version of the estimated direct link, $\bar{\boldsymbol{y}}^{(n)}_{j,k} = \tilde{\boldsymbol{y}}^{(n)}_{j,k} - A_k\hat{\boldsymbol{h}}^{\mathrm{H}}_{d,k}$. Concatenating the $N$ entries as rows yield $\bar{\boldsymbol{Y}}_{j,k} = A_k \boldsymbol{V}\boldsymbol{H}_{j,k} +  A_k\boldsymbol{1}_N \boldsymbol{e}_{d,k}^{\mathrm{H}} + \bar{\boldsymbol{Z}}_{j,k}$. 
Using unitary property $\boldsymbol{V}^{\mathrm{H}}\boldsymbol{V}=N\mathbf{I}_N$,   we compute $\boldsymbol{Y}^\prime_{j,k}=\boldsymbol{V}^{\mathrm{H}}\bar{\boldsymbol{Y}}_{j,k}/N$. Hereafter, we denote row $n$ of any enclosed matrix using $(.)^{(n)}$. Using $\boldsymbol{V}^{\mathrm{H}}\boldsymbol{1}_N = [N,0,\ldots,0]^{\mathrm{T}} \in \mathbb{R}^{N\times1}$ we get row $n$ of $\boldsymbol{Y}^\prime_{j,k}$ as 
\begin{align}
\boldsymbol{Y}_{j,k}^{\prime(n)}
&=
A_k\boldsymbol{H}_{j,k}^{(n)}
+
\boldsymbol{1}_{\{n=1\}}A_k\boldsymbol{e}_{d,k}^{\mathrm{H}}
+
\left(\frac{\boldsymbol{V}^{\mathrm{H}}}{N}
\bar{\boldsymbol{Z}}_{j,k}\right)^{(n)}.
\label{eq:lmmseobs}
\end{align}
for any $n\in\{1,2,\dots,N\}$. This is observed sample that we use to estimate the $n^{\mathrm{th}}$ row of the indirect channel. Its mean is $\mathbb{E}[\boldsymbol{Y}_{j,k}^{\prime(n)}]
=
A_k\bar{\boldsymbol{H}}_{j,k}^{(n)}$. The covariance terms are given below
\begin{align}
&\boldsymbol{R}_{H_{j,k}^{(n)}H_{j,k}^{(n)}}
=
\mathbb{E}\left[
(\boldsymbol{H}_{j,k}^{(n)}-\bar{\boldsymbol{H}}_{j,k}^{(n)})^{\mathrm{H}}
(\boldsymbol{H}_{j,k}^{(n)}-\bar{\boldsymbol{H}}_{j,k}^{(n)})
\right]
\nonumber\\
&=
\frac{\beta_{\mathrm{in}}}{(1+T_1)(1+T_2)}
\left[
(1+T_2)\boldsymbol{I}_M
+
\frac{T_1}{N}
\bar{\boldsymbol{H}}_{b,j}^{\mathrm{H}}
\bar{\boldsymbol{H}}_{b,j}
\right],
\label{eq:rescov}
\end{align}
\begin{align}
&\boldsymbol{R}_{Y^{\prime(n)}_{j,k}H^{(n)}_{j,k}}
=
\mathbb{E}\left[
\left(
\boldsymbol{Y}_{j,k}^{\prime(n)}
-\mathbb{E}[\boldsymbol{Y}_{j,k}^{\prime(n)}]
\right)^{\mathrm{H}}
\left(
\boldsymbol{H}_{j,k}^{(n)}
-\bar{\boldsymbol{H}}_{j,k}^{(n)}
\right)
\right],
\nonumber\\
&=
A_k\boldsymbol{R}_{H^{(n)}_{j,k}H^{(n)}_{j,k}},
\end{align}
\begin{align}
&\boldsymbol{R}_{Y^{\prime(n)}_{j,k}Y^{\prime(n)}_{j,k}}
=
\mathbb{E}\left[
\left(
\boldsymbol{Y}_{j,k}^{\prime(n)}
-\mathbb{E}[\boldsymbol{Y}_{j,k}^{\prime(n)}]
\right)^{\mathrm{H}}
\left(
\boldsymbol{Y}_{j,k}^{\prime(n)}
-\mathbb{E}[\boldsymbol{Y}_{j,k}^{\prime(n)}]
\right)
\right]
\nonumber\\
&=
A_k^2\boldsymbol{R}_{H^{(n)}_{j,k}H^{(n)}_{j,k}}
+
\left(
\boldsymbol{1}_{\{n=1\}}A_k^2\sigma_{e,k}^2
+
\frac{\sigma_{\bar{Z},k}^2}{N}
\right)\boldsymbol{I}_M .
\end{align}
Thus, the LMMSE estimate of the $n^{\mathrm{th}}$ row of
$\boldsymbol{H}_{j,k}$ is
\begin{align}
\hat{\boldsymbol{H}}_{j,k}^{(n)}
=
\bar{\boldsymbol{H}}_{j,k}^{(n)}
+
\left(
\boldsymbol{Y}_{j,k}^{\prime(n)}
-
A_k\bar{\boldsymbol{H}}_{j,k}^{(n)}
\right)
\boldsymbol{R}^{-1}_{Y^{\prime(n)}_{j,k}Y^{\prime(n)}_{j,k}}
\boldsymbol{R}_{Y^{\prime(n)}_{j,k}H^{(n)}_{j,k}} .
\label{eq:indirectest}
\end{align}
This process is repeated sequentially for all $J$ RISs, providing the full end-to-end estimated channel $\boldsymbol{h}_k = \hat{\boldsymbol{h}}_k + \boldsymbol{e}_k$.

\subsection{Joint Beamforming}
\label{subsec:jointbeam}

Once we have estimated CSI, we can solve $\mathbf{P0}$ in real-time. To begin with, we use the Lagrangian dual transform \cite{fp2}, to reduce $\mathbf{P0}$ into a sum-of-ratios problem by introducing a set of auxiliary variables $\boldsymbol{\alpha} = \{\alpha_1,\dots,\alpha_K\}$ 

\begin{subequations}
\begin{align}
\mathbf{P0(a)}:\quad\max_{\boldsymbol{\alpha},\boldsymbol{W},\boldsymbol{\Theta}}&  \sum_{k=1}^K\Big[ \log(1+\alpha_k) - \alpha_k + \frac{(1+\alpha_k)\gamma_k}{1+\gamma_k}\Big]\nonumber \\
s.t. \quad &\alpha_k\ge0, \quad \forall k,\\
& \left|\theta_j^{(n)}\right| = 1, \quad \forall \ j, n,\\
&\sum_{k=1}^{K}\|\boldsymbol{w}_k\|_2^2 \leq P_T.
\end{align}
\end{subequations}

% \begin{equation}
% \max_{\boldsymbol{\alpha}, \boldsymbol{\Theta}, \boldsymbol{W}} \sum_{k=1}^{K} \log(1 + \alpha_k) - \alpha_k + \frac{(1 + \alpha_k)\gamma_k}{1 + \gamma_k}
% \end{equation}

% \begin{itemize}
%     \item[$\circ$] $s.t. \quad |\theta_g^{(n)}| = 1, \ \forall \ 1 \le g \le G, \ 1 \le n \le N$
%     \item[$\circ$] $\sum_{k=1}^{K} |\boldsymbol{w}_k|^2 \le P_T$
% \end{itemize}

% \begin{equation}
% \begin{aligned}
%     \frac{\partial f}{\partial \alpha_k} & \Rightarrow \frac{1}{1 + \alpha_k} - 1 + \frac{\gamma_k}{1 + \gamma_k} = 0 \\
%     & \Rightarrow \frac{1 + \gamma_k - (1 + \alpha_k)(1 + \gamma_k) + \gamma_k(1 + \alpha_k)}{(1 + \alpha_k)(1 + \gamma_k)} = 0 \\
%     & \Rightarrow \frac{1 + \gamma_k - 1 - \alpha_k\gamma_k - \alpha_k - \gamma_k + \gamma_k + \alpha_k\gamma_k}{(1 + \alpha_k)(1 + \gamma_k)} = 0 \\
%     & \Rightarrow \alpha_k^{opt} = \gamma_k
% \end{aligned}
% \end{equation}

Clearly, the problem is equivalent to $\mathbf{P0}$ as long as $\alpha_k = \gamma_k$. Then, the quadratic transform \cite{fp1} is applied to decouple the ratios (SINR terms) in $\mathbf{P0(a)}$ by introducing another set of auxiliary variables $\boldsymbol{\beta} = \{\beta_1,\dots,\beta_K\}$
\begin{subequations}
\begin{align}
\mathbf{P0(b)}:\quad\max_{\boldsymbol{\alpha},\boldsymbol{\beta},\mathbf{W},\boldsymbol{\Theta}}&  \sum_{k=1}^K\Big[ (\log(1+\alpha_k) - \alpha_k) \nonumber\\
+ & 2\sqrt{(1+\alpha_k)} \mathcal{R}e\Big\{ \beta_k^*\hat{\boldsymbol{h}}_k\boldsymbol{w}_k \Big\} \nonumber\\
- &|\beta_k|^2 \Big( \sum_{i=1}^K |\hat{\boldsymbol{h}}_k\boldsymbol{w}_i|^2 + \sigma^2 \Big)\Big]\nonumber \\
s.t. \quad &\alpha_k\ge0, \quad \forall k,\\
&|\theta^{(n)}_j|=1, \forall j,n,\\
&\sum_{k=1}^{K}\|\boldsymbol{w}_k\|_2^2 \leq P_T.
\end{align}
\end{subequations}
\begin{table*}[h!]
\caption{Computational Complexity Analysis of the Proposed RIS Placement.}
\label{table_complexity}
\centering
% 1. Increase vertical padding (1.5 is usually the "sweet spot" for math)
\renewcommand{\arraystretch}{1.5} 
% 2. Increase horizontal padding between columns
\setlength{\tabcolsep}{10pt} 

\begin{tabular}{|c|l|c|}
\hline
\textbf{Phase} & \textbf{Step} & \textbf{Complexity} \\
\hline
\multirow{5}{*}{\textbf{Deployment}} 
& Discretization & $\mathcal{O}(|\mathcal{S}|)$ \\ \cline{2-3}
& DBSCAN & $\mathcal{O}(|\mathcal{Q}|\log|\mathcal{Q}|)$  \\ \cline{2-3}
& Construct $G_{\mathrm{vis}},\mathcal{V}$ + Proposition \ref{prop:1} & $\mathcal{O}(|(\mathcal{S}^*\cup\mathcal{V})\setminus\mathcal{Q}||\mathcal{S}^*\setminus\mathcal{Q}| + |\mathcal{V}||\mathcal{U}|) \approx\mathcal{O}(|(\mathcal{S}^*\cup\mathcal{V})\setminus\mathcal{Q}||\mathcal{S}^*\setminus\mathcal{Q}|)$ \\ \cline{2-3}
& Construct $G_{\mathrm{int}}$ + Algorithm~\ref{algo:greedy} & $\mathcal{O}(T(2^T + T|\mathcal{C}^{\prime}_{\max}|)) $ \\ \cline{2-3}
& BO-based RIS Placement & $\mathcal{O}(J \cdot I_{BO}(C_{CE} + C_{JB}))$ \\
\hline
\multirow{2}{*}{\textbf{Operation}} 
& Channel Estimation & $ \mathcal{O}(C_{CE} = J K NM\log N )$ \\ \cline{2-3}
& Joint Beamforming & $\mathcal{O}(C_{JB} = I_{JB} \cdot K(M^2 + JN^2\min\{M,K\}))$ \\
\hline
\multicolumn{2}{|c|}{\textbf{Deployment Overall}} & $\mathcal{O}(|(\mathcal{S}^*\cup\mathcal{V})\setminus\mathcal{Q}||\mathcal{S}^*\setminus\mathcal{Q}| + 2J \cdot I_{BO}(C_{CE} + C_{JB}))$ \\
\hline
\multicolumn{2}{|c|}{\textbf{Operation Overall}} & $\mathcal{O}(C_{CE} + C_{JB})$ \\
\hline
\end{tabular}
\end{table*}

In order to solve $\mathbf{P0(b)}$, we adopt an alternating optimization (AO) strategy. The auxiliary variables are updated with optimal values by setting their partial derivatives to zero \cite{beam2} 
\begin{align}
\alpha_k &= \gamma_k,\label{eq:a}\\
\beta_k &= \frac{\sqrt{ (1+\alpha_k)}\,\hat{\boldsymbol{h}}_k \boldsymbol{w}_k}{\sum_{i=1}^K |\hat{\boldsymbol{h}}_k \boldsymbol{w}_i|^2 + \sigma^2}.
\label{eq:b}
\end{align}
% where
% \begin{equation}
% \beta_k^{\text{opt}} = \frac{\sqrt{(1 + \alpha_k)} \boldsymbol{h}_k^{\mathrm{H}} \boldsymbol{w}_k}{\sigma^2 + \sum_{i=1}^{K} |\boldsymbol{h}_k^{\mathrm{H}} \boldsymbol{w}_i|^2}
% \end{equation}
Looking at the well-studied optimal transmit beamforming problem with constant phases, one can derive the update rule for $\mathbf{W}$ to be the Weighted Minimum Mean Square Error (WMMSE) solution which is given by
\begin{align}
\boldsymbol{w}_k &=  \sqrt{(1 + \alpha_k)} \beta_k \Big(\kappa^{\text{opt}} \boldsymbol{I}_M +\sum_{i=1}^K |\beta_i|^2 \hat{\boldsymbol{h}}_i \hat{\boldsymbol{h}}_i^{\mathrm{H}} \Big)^{-1} \hat{\boldsymbol{h}}_k,\label{eq:W}\\
% Equation 2
\kappa^{\text{opt}} &= \min \left\{ \kappa^{\text{opt}} \geq 0 \mid \sum_{i=1}^K \|\boldsymbol{w}_k\|_2^2 \le P_T \right\},
\label{eq:lambda}
\end{align}
where $\kappa^{\text{opt}}$ is obtained by solving~\eqref{eq:lambda} using a bisection search over a finite interval from $0$. Finally, retaining the phase terms in  $\mathbf{P0(b)}$ and grouping quadratic and linear terms \cite{loc2}, 
\begin{align}
\mathbf{P0(c):} \quad    &\min_{\boldsymbol{\theta}_j} f_3(\boldsymbol{\Theta}) = \sum_{j=1}^J(\boldsymbol{\theta}_j^{\mathrm{H}}\boldsymbol{U}_j\boldsymbol{\theta}_j - 2 \mathcal{R}e\{\boldsymbol{v}_j^{\mathrm{H}}\boldsymbol{\theta}_j\})\\
    \text{where }\quad &\boldsymbol{U}_j = \sum_{k=1}^{K} |\beta_k|^2 \sum_{i=1}^{K} \hat{\boldsymbol{H}}_{j,k} \boldsymbol{w}_i \boldsymbol{w}_i^{\mathrm{H}} \hat{\boldsymbol{H}}_{j,k}^{\mathrm{H}},\\
    \boldsymbol{v}_j &= \sum_{k=1}^{K} \sqrt{(1 + \alpha_k)} \beta_k^* (\hat{\boldsymbol{H}}_{j,k}\boldsymbol{w}_k)\nonumber \\
&- |\beta_k|^2 \sum_{i=1}^{K} ((\hat{\boldsymbol{h}}_{d,k}^{\mathrm{H}} + \sum_{\substack{\ell=1\\ \ell\ne j}}^{J} \boldsymbol{\theta}_\ell^T \hat{\boldsymbol{H}}_{\ell,k})  \boldsymbol{w}_i)^* \hat{\boldsymbol{H}}_{j,k} \boldsymbol{w}_i.
\end{align}
% \begin{equation}
% 
% 1\boldsymbol{U}_g = \sum_{k=1}^{K} |\beta_k|^2 \sum_{i=1}^{K} \boldsymbol{H}_{g,k} \boldsymbol{w}_i \boldsymbol{w}_i^{\mathrm{H}} \boldsymbol{H}_{g,k}^{\mathrm{H}}
% \end{equation}
% \begin{align}
% \boldsymbol{v}_g &= \sum_{k=1}^{K} \sqrt{(1 + \alpha_k)} \beta_k^* (\boldsymbol{H}_{g,k}\boldsymbol{w}_k)\nonumber \\
% &- |\beta_k|^2 \sum_{i=1}^{K} ((\boldsymbol{h}_{d,k}^{\mathrm{H}} + \sum_{j\ne g} \boldsymbol{\theta}_g^{\mathrm{H}} \boldsymbol{H}_{g,k})  \boldsymbol{w}_i)^* \boldsymbol{H}_{g,k} \boldsymbol{w}_i
% \end{align}

Hence, to optimize phases with constant transmit beamforming, we adopt the gradient projection algorithm 
\begin{equation}
\boldsymbol{\theta}_j^{(t+1)} = \left( \boldsymbol{\theta}_j^{(t)} - \rho \nabla f_3 \left( \boldsymbol{\theta}_j^{(t)} \right) \right)^+
\label{eq:phase}
\end{equation}
where $(.)^+$ denotes the projection of the enclosed set of phases into the region $|\theta^{(n)}_j| =1, \forall j,n$. Also, $\rho$ is dynamically updated in each iteration for faster convergence. Since $\boldsymbol{v}_j$ has cross terms $\{\boldsymbol\theta_\ell\}_{\ell \neq j}$, we adopt a sequential block coordinate descent -- at each outer iteration, RISs are updated in order $j = 1,\dots,J$, with $\boldsymbol{v}_j$ using the recently available $\{\boldsymbol\theta_\ell\}_{\ell \neq j}$.

We alternatively optimize $(\boldsymbol{\alpha},\boldsymbol{\beta},\mathbf{W},\boldsymbol{\Theta})$ using~\eqref{eq:a},~\eqref{eq:b},~\eqref{eq:W},~\eqref{eq:phase} iteratively to solve the subproblem.

\begin{remark}
Although in Sections \ref{subsec:imperfectcsi} and \ref{subsec:jointbeam}, we have listed a variation of the most well-accepted channel estimation, RIS training, and joint beamforming techniques~\cite{channelestigen, weightedsumrate}, owing to the separable structure of the problem formulation $\mathbf{P0}-\mathbf{P2}$, the algorithm remains effective for any proposed or future channel estimation, RIS training techniques~\cite{channelesti1, channelestisurvey1, channelestisurvey2, channelesti2, channelestimulti} and any joint beamforming technique~\cite{loc2, beam2,beam0}. %the proposed approach remains \emph{effective} for any other or future CSI acquisition, RIS training, and joint beamforming technique.   
\end{remark}%, beam4, beam5

\subsection{User Distribution driven multi-dimensional BO for Optimal RIS placement}

Finally, to solve problem $\mathbf{P2}$, we consider joint optimization of placement of $J$ RISs to maximize expected sum rate. The definition of optimal rate, $R^*$
defines that the RIS coefficients and transmit beamformers are to be obtained by solving $\mathbf{P0}$. Since closed-form solutions for $\boldsymbol{W}$ and $\boldsymbol{\Theta}$ as functions of the channel is intractable, obtaining an exact solution to $\mathbf{P2}$ is practically impossible. Therefore, we are instead required to approximate the expectation over multiple sets of instantiations of users from the user distribution and for each of those instantiations, over multiple instantiations of the Rayleigh/Rician channels to approximate the two-layer expectation. 

There are multiple possibilities like Monte Carlo or recursive clustering~\cite{conf}, etc. Due to the high computational cost of expected sum rate (joint beamforming and averaging over user distribution and random channels) and the non-convex nature of the search space, we employ a $2J$-dimensional BO~\cite{bayesopt} for jointly searching for the best position for each RIS $i$, for $1\le i\le J$ from each candidate set $C_i$. In the BO, we assume a Gaussian surrogate to approximate the objective function and use the expected improvement (EI) acquisition function to balance exploitation and exploration. The optimization is performed for a large number of user instantiations and channel realizations. $J$ is obtained from the method listed in Section \ref{sec:ss}, satisfying~\eqref{subeq:J}. We use joint beamforming for calculating the optimal phases and beamformers  to calculate $R^*$ in the BO. All locations in each candidate set $C_i$ directly satisfy~\eqref{subeq:coverage2} and~\eqref{subeq:BS2} by construction. 

% We believe that this idea of combining user data-driven learning with conventional beamforming to estimate and hence maximize the expected post-deployment throughput of the system to intelligently deploy is a practical and novel one.

\subsection{Complexity Analysis}
\label{subsec:complexity}
A detailed summary of the time complexity of each step of the proposed method is listed in Table \ref{table_complexity}. DBSCAN as in~\cite{dbscan} and in the implementation available across programming languages has complexity $\mathcal{O}(|\mathcal{Q}|\log(|\mathcal{Q}|))$. The construction of $G_{\mathrm{vis}}$ and $\mathcal{V}$ is $\mathcal{O}(|(\mathcal{S}^*\cup\mathcal{V})\setminus\mathcal{Q}||\mathcal{S}^*\setminus\mathcal{Q}|)$ which is the costliest operation before BO. Then, construction of preliminary candidate sets is $\mathcal{O}(|\mathcal{V}||\mathcal{U}|)$. The redundancy reduction and final candidate set construction is $\mathcal{O}(T 2^T +T^2|\mathcal{C}^{\prime}_{\mathrm{max}}|)$ where $\mathcal{C}^{\prime}_{\mathrm{max}}$ is the largest preliminary candidate set. Though BO has multiple steps, sum rate evaluation clearly outweighs all others as a function of the complexity of the chosen channel estimation $C_{CE}$ and joint beamforming technique $C_{JB}$. We provide expansions for the listed algorithms. Channel estimation is dominated by the $JK$ indirect links. The Fast Fourier Transform (FFT) implementation of multiplication of DFT matrices reduces complexity from $\mathcal{O}(N^2 M)$ to $\mathcal{O}(MN\log(N))$. We also use $\text{Rank}=1$ nature of the steering vector term of  $\boldsymbol{R}_{H^{(n)}_{j,k}H^{(n)}_{j,k}}$ to avoid $\mathcal{O}(M^3)$.  Efficient implementation and complexity discussion for the joint beamforming algorithm is in \cite{conf}. Thus, the overall complexity of the proposed one-time RIS deployment is dominated by the construction of $V_{\mathrm{int}}$ and joint beamforming. Let $I_{BO}$ and $I_{JB}$ respectively indicate the total number of BO and joint beamforming iterations. 
\begin{remark}
Number of iterations required to obtain a good placement using BO is much lesser than Monte Carlo or recursive clustering \cite{conf}. Choosing to form `large obstacles' instead of looking at each obstacle ($I^2$ to $T^2$) and considering only those grid points that aren't obstacles, aren't already covered, and have a positive user intensity in $\mathcal{U}=\mathcal{S}^*\setminus \mathcal{Q} \setminus \mathcal{V}$ reflects a clear reduction in complexity. 
\end{remark}

\begin{figure*}
   %---------- ROW 1 ----------
    \begin{subfigure}[b]{0.49\columnwidth}
        \includegraphics[width=\linewidth]{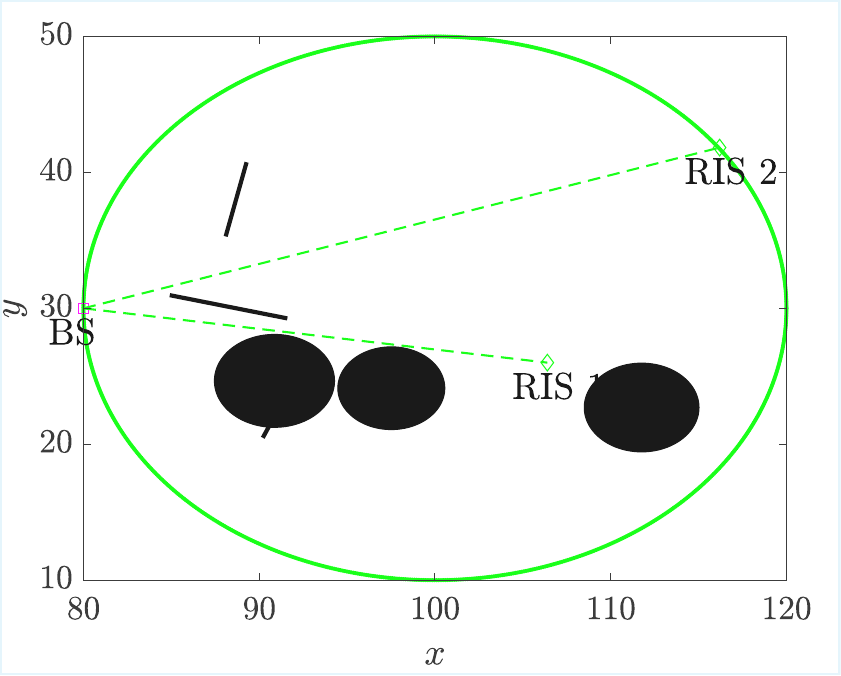}
        \caption{Scenario 1.}
        \label{fig:sc1}
    \end{subfigure}
    \hfill
     % Pushes the next subfigure to the right
    \begin{subfigure}[b]{0.49\columnwidth}
        \includegraphics[width=\linewidth]{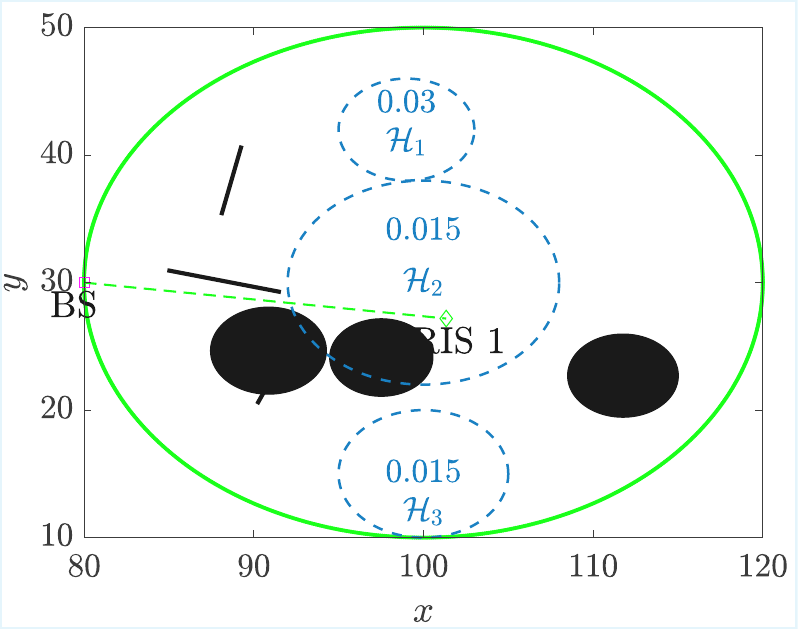}
        \caption{Scenario 2.}
        \label{fig:sc2}
        \end{subfigure}
\hfill
    %---------- ROW 2 ----------
    \begin{subfigure}[b]{0.49\columnwidth}
        \includegraphics[width=\linewidth]{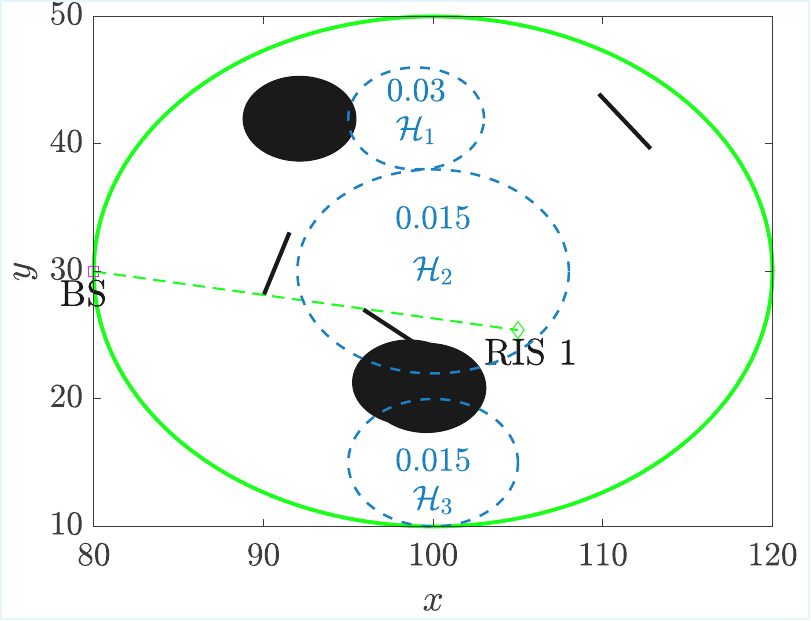}
        \caption{Scenario 3.}
        \label{fig:sc3}
    \end{subfigure}
    \hfill % Pushes the next subfigure to the right
    \begin{subfigure}[b]{0.49\columnwidth}
        \includegraphics[width=\linewidth]{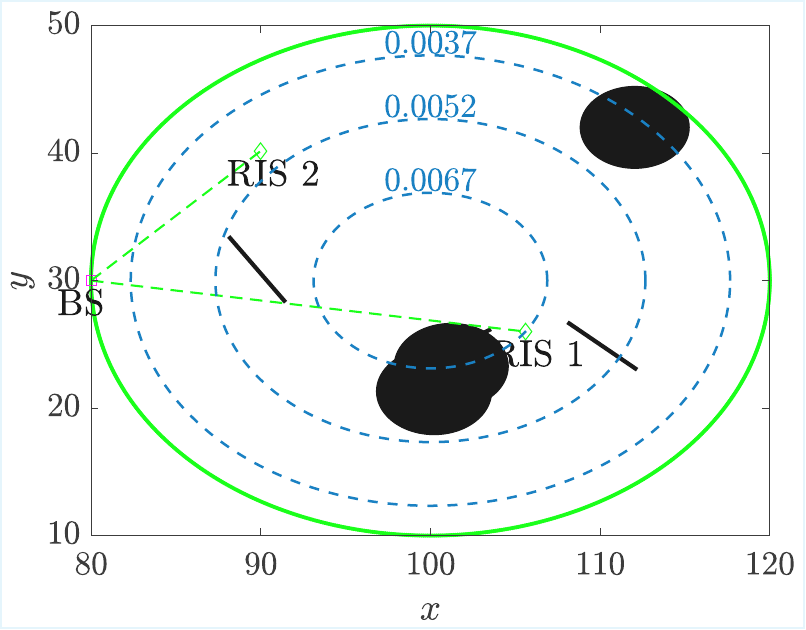}
        \caption{Scenario 4.}
        \label{fig:sc4}
    \end{subfigure}
     \caption{Obstacle Configurations, User distributions, and RIS placement in all scenarios.}
     \label{fig:allscenariosfina}
\end{figure*}
\begin{figure*}
    \begin{subfigure}[b]{0.51\columnwidth}
        \includegraphics[width=\linewidth]{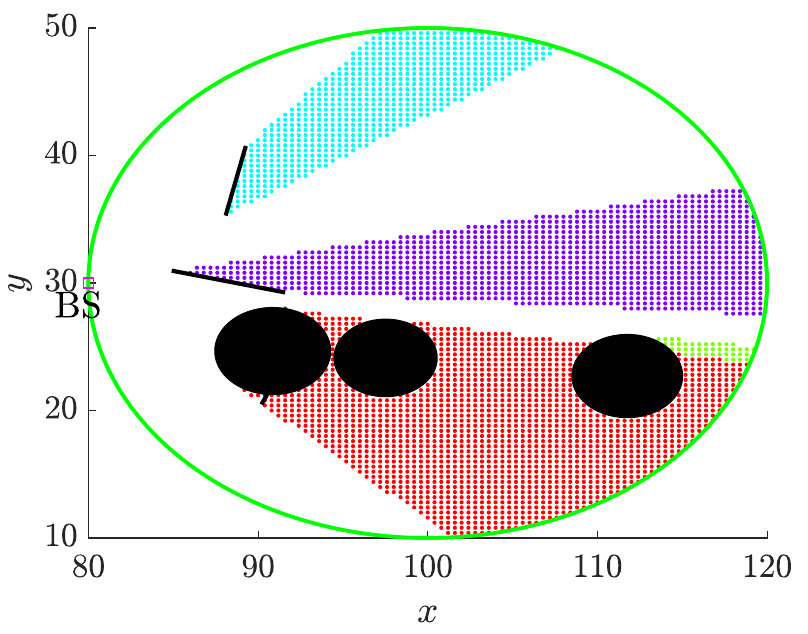}
        \caption{Obstacle Clusters.}
        \label{fig:pcs0}
    \end{subfigure}
    \hfill
   %---------- ROW 1 ----------
    \begin{subfigure}[b]{0.45\columnwidth}
        \includegraphics[width=\linewidth]{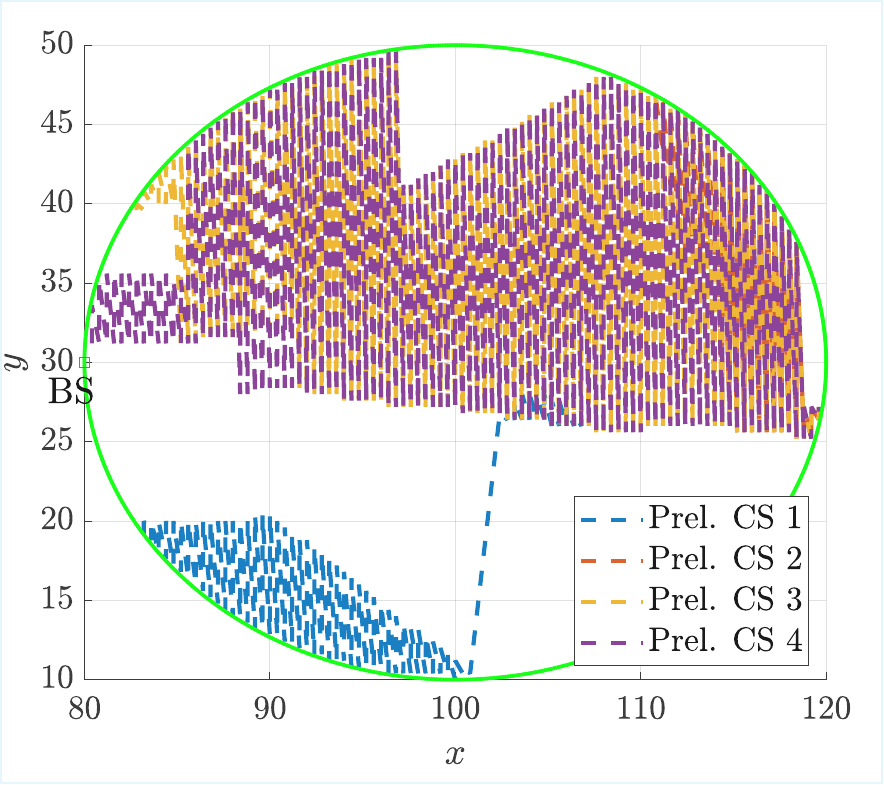}
        \caption{Preliminary Candidate Sets.}
        \label{fig:pcs1}
    \end{subfigure}
    \hfill
     % Pushes the next subfigure to the right
    \begin{subfigure}[b]{0.49\columnwidth}
        \includegraphics[width=\linewidth]{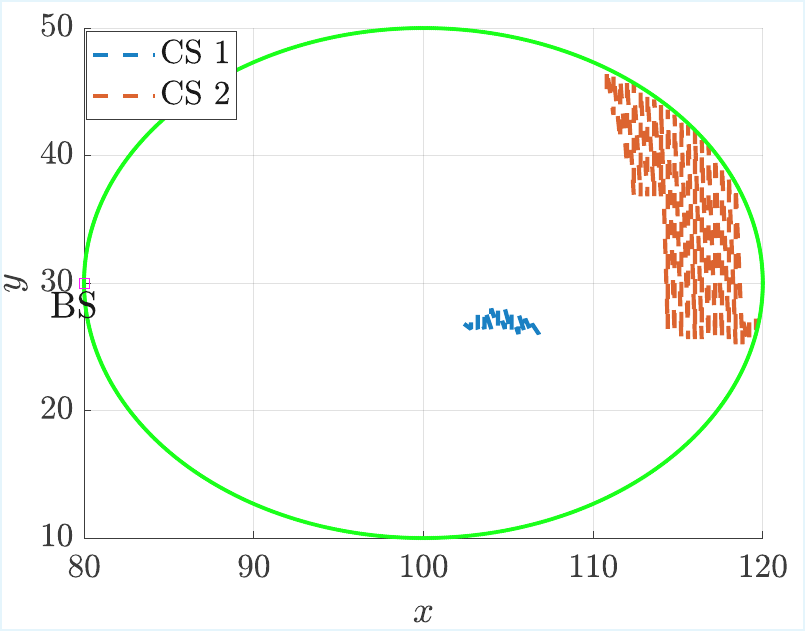}
        \caption{Final Candidate Sets.}
        \label{fig:fcs}
        \end{subfigure}
     \hfill
    \begin{subfigure}[b]{0.52\columnwidth}
        \includegraphics[width=\linewidth]{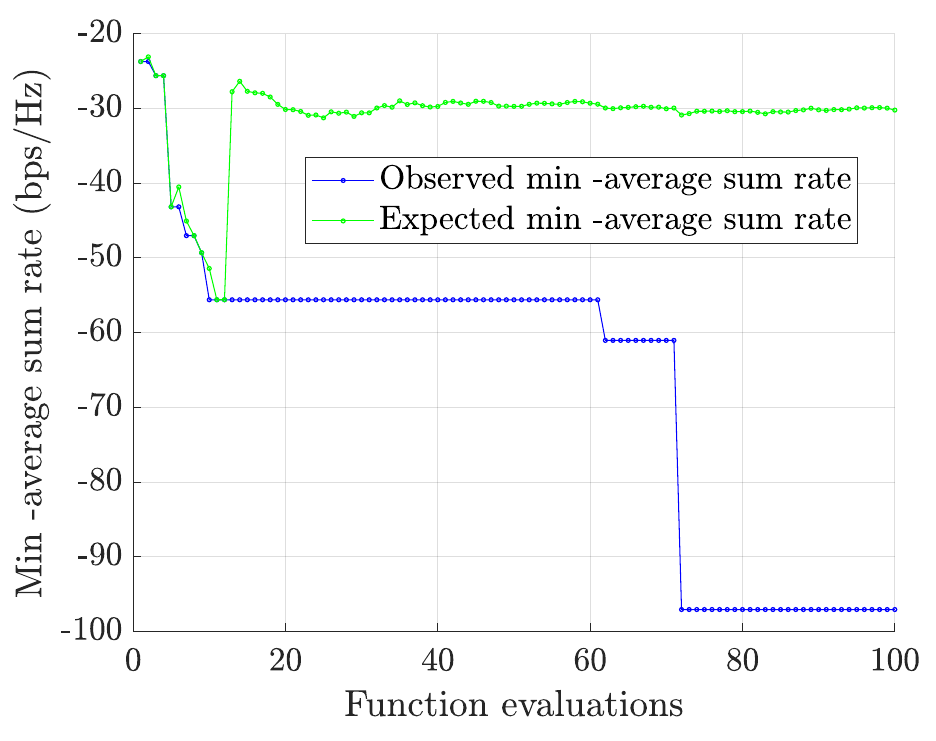}
        \caption{Convergence Behavior.}
        \label{fig:pcs3}
    \end{subfigure}
     \caption{Proposed Algorithm through Scenario 1.}
     \label{fig:procedure}
\end{figure*}
\begin{table*}[h!]
\caption{User Distributions.}
\label{table_distributions}
\centering
\begin{tabular}{|c|c|c|}
\hline
\textbf{Distribution} & \textbf{Definition} & \textbf{Parameters} \\
\hline
$f_1$ & $\lambda(\mathbf{s}) = \lambda_0$, $\forall \mathbf{s} \in \mathcal{S}\setminus\mathcal{Q}$ & $\lambda_0 = 0.003$ users/m$^2$ \\
\hline
$f_2$ & $\lambda(\mathbf{s}) = \begin{cases} \lambda_i & \mathbf{s} \in \mathcal{H}_i\setminus\mathcal{Q},\ i=1,2,3 \\ 0 & \text{otherwise} \end{cases}$ & $\mathcal{H}_1 = \text{Circle}((99, 42),\ 4)$, $\lambda_1 = 0.03$ users/m$^2$ \\
 & where $\mathcal{H}_i = \text{Circle}(\mathbf{c}_i, r_i)$ is a disk of & $\mathcal{H}_2 = \text{Circle}((100, 30),\ 8)$, $\lambda_2 = 0.015$ users/m$^2$\\
 & radius $r_i$ centered at $\mathbf{c}_i$ & $\mathcal{H}_3 = \text{Circle}((100, 15),\ 5)$, $\lambda_3 = 0.015$ users/m$^2$\\
\hline
$f_3$ & $\lambda(\mathbf{s}) = \lambda_0 \exp\!\left( -\dfrac{\|\mathbf{s} - \boldsymbol{\mu}\|^2}{2\sigma^2} \right), \forall\boldsymbol{s}\in\mathcal{S}\setminus\mathcal{Q}$ & $\lambda_0 = 7.5 \times 10^{-3}$ users/m$^2$ \\
 & & $\boldsymbol{\mu} = (100, 30)$~m, $\sigma = 15$~m \\
\hline
\end{tabular}
\end{table*}

\section{Results and Discussion}
\label{sec:results}
The path loss parameters are assumed to obey the 3GPP standard for Urban Micro cell (Table B.1.2.1-1 in \cite{3gpp}),
\begin{align}   
    &\beta_{{i}_{\text{dB}}} = 22\log_{10}(d_i) + 28 + 20\log_{10}(f_c) \text{, }i\in\{(b,j),(j,k)\},\\
    &\beta_{{d,k}_{\text{dB}}} = 36.7\log_{10}(d_{d,k}) + 22.7 + 26\log_{10}(f_c),
\end{align}
where $d_i$ represents the Euclidean distance traversed by each link in meters (m) and $f_c$ is the carrier frequency in GHz. $f_c$ is set as $5$ GHz. For example, we get a loss of $\beta_{{b,j}_{\text{dB}}} = \beta_{{j,k}_{\text{dB}}} = 63.98$ dB and $\beta_{{d,k}_{\text{dB}}} = 77.57$ dB for $d_{b,j} = d_{j,k} = d_{d,k} =  10$~m. Then the Rician factors are set to $T_{1_{\text{dB}}} = T_{2_{\text{dB}}} = 10$ dB. The transmit power is set to $P_{T_{\text{dB}}} = 30$ dBm and noise level $\sigma^2_{\text{dB}} = -100$ dBm \cite{pathloss} unless mentioned otherwise. Across all simulations, other system parameters are set to $M = 16, N = 100, \epsilon=1,N_{\mathrm{DBSCAN}}=5$. We choose the default $p_{c_{th}}=0.75$.

\subsection{Obstacle Configurations and User distributions}
The BS is assumed to be at $\boldsymbol{o}_b = (80$ m$,30$ m$)$ at a corner of the cell as in Fig. \ref{fig:sysmodel} and the cell radius is fixed as $R = 20$~m. 
To analyze the proposed algorithm under various types of environments, we consider 3 different obstacle configurations as depicted in Fig.~\ref{fig:allscenariosfina}. We also consider three user distributions, i.e., $f_1$, $f_2$ and $f_3$ as defined in Table \ref{table_distributions}. $f_1$ is a uniform distribution over the whole cell. $f_2$ and $f_3$ both are inhomogeneous distributions -- where the former involves three user hotspots and the latter a  multi-dimensional Gaussian. It is notable that the densities are defined across distributions such that the average number of users in the cell remains $5$. Based on the user distributions and obstacle configurations, four scenarios are defined as in Table~\ref{table_results} and Fig.~\ref{fig:allscenariosfina}.  

% \begin{figure*}
%    %---------- ROW 1 ----------
%     \begin{subfigure}[b]{0.51\columnwidth}
%         \includegraphics[width=\linewidth]{Figures/Results/2.CS_cropped.pdf}
%         \caption{Scenario 1.}
%         \label{fig:scenario1}
%     \end{subfigure}
%     \hfill
%      % Pushes the next subfigure to the right
%     \begin{subfigure}[b]{0.50\columnwidth}
%         \includegraphics[width=\linewidth]{Figures/Results/Scenario2_CS_cropped.pdf}
%         \caption{Scenario 2.}
%         \label{fig:scenario2}
%         \end{subfigure}
% \hfill
%     %---------- ROW 2 ----------
%     \begin{subfigure}[b]{0.49\columnwidth}
%         \includegraphics[width=\linewidth]{Figures/Results/Scenario3_CS_cropped.pdf}
%         \caption{Scenario 3.}
%         \label{fig:scenario3}
%     \end{subfigure}
%     \hfill % Pushes the next subfigure to the right
%     \begin{subfigure}[b]{0.51\columnwidth}
%         \includegraphics[width=\linewidth]{Figures/Results/Scenario4_CS_cropped.pdf}
%         \caption{Scenario 4.}
%         \label{fig:scenario4}
%     \end{subfigure}
%      \caption{Results across scenarios for Search Space Determination and Minimization of number of RISs.}
%      \label{fig:allscenarios}
% \end{figure*}

\begin{table}[h!]
\caption{Results across Scenarios.}
\label{table_results}
\centering
\resizebox{\columnwidth}{!}{\begin{tabular}{|c|c|c|c|c|c|c|}
\hline
\textbf{Scenario} & \textbf{Obstacle} & \textbf{User} & \textbf{J} & \textbf{Optimal RIS Locations} & \textbf{Sum rate} & \textbf{Coverage} \\
 & \textbf{Config.} & \textbf{Dist.} & & \textbf{(m,m)} & \textbf{(bps/Hz)} & \textbf{(\%)} \\
\hline
$1$ & $1$ & $f_1$ & $2$ & $(106.40, 26.00)$ & $29.2375$ & $96.52$ \\
 &  &  &  & $(116.20, 41.83)$ &  &  \\
$2$ & $1$ & $f_2$ & $1$ & $(101.34, 27.18)$                & $28.0848$ & $94.96$ \\
3 & 2 & $f_2$ & 1 & $(105.00, 25.40)$                & $28.2473$ & $98.83$ \\
4 & 3 & $f_3$ & 2 & $(105.61, 26.00)$                & $30.2920$ & $99.28$ \\
 &  &  &  & $(89.99, 40.15)$ &  &  \\
%$3$ & $2$ & $f_2$ & $1$ & $(106.77, 25.99)$                & $29.4564$ & $99.22$ \\
\hline
\end{tabular}}
\end{table}

\subsection{Performance metrics}
Across this section, \textit{Average sum rate} denotes the sample mean of the sum rate of system over multiple instantiations of users and the random channel. \textit{Coverage} $\bar{p}_c$ is given by
\begin{equation}
    \bar{p}_c = \frac{\sum_{\boldsymbol{s} \in \mathcal{S} \setminus \mathcal{Q}} \Lambda(\boldsymbol{s})\,\boldsymbol{1}_{\{L_{\mathrm{net}}(\boldsymbol{s})\ge 1\}}}{\sum_{\boldsymbol{s} \in \mathcal{S} \setminus \mathcal{Q}} \Lambda(\boldsymbol{s})},
\end{equation}
with a grid resolution of $0.4$~m. This serves as a standard empirical estimate for our definition of $p_{\mathrm{cov}}$ \cite{loc1}. Since we assume solid obstacles, these points cannot be extended coverage, so these are omitted in the above equation.%($\tau_{ov} = 0$)

By default, simulations assume perfect CSI. For imperfect CSI simulations, the degree of imperfection is quantified by Normalized Mean Square Error (NMSE),
\begin{equation}
    \varrho = \frac{\mathbb{E}\left[|\boldsymbol{h} - \hat{\boldsymbol{h}}|^2\right]}{\mathbb{E}\left[|\boldsymbol{h}|^2\right]},
\end{equation}
where $\boldsymbol{h}$ and $\hat{\boldsymbol{h}}$ denote the true and estimated value of both $\boldsymbol{h}_d$ and $\mathbf{H}_{j,k}$, respectively. For a fair comparison across all simulations, we assume that the channel remains constant over a sufficiently large coherence interval  ($\tau_{tot}\rightarrow\infty$). 
% \begin{remark}
% Though we present results for $\tau_{tot}\rightarrow \infty$, since the factor $(1-\frac{\tau_{ov}}{\tau_{tot}})$ is independent of both $\boldsymbol{W}$ and $\boldsymbol{\Theta}$ in $\mathbf{P0}$, the proposed algorithm holds even when $\tau_{tot}$ is finite.
% \end{remark}
\begin{remark}
Though a Gaussian error of $\varrho$ is introduced to each direct and indirect link, the realized net NMSE of the cascaded 
BS-RIS-user channel $\boldsymbol{H}_{j,k}$ may slightly differ from $\varrho$ due to error 
propagation from the direct-link estimation.
\end{remark}
\subsection{Illustration of the proposed method using Scenario 1}
We illustrate the procedure of the proposed method in Fig. \ref{fig:procedure} for Scenario 1. Fig. \ref{fig:pcs0} illustrates the discretization of the space and the regions shadowed by ``large obstacles'', $\{\mathcal{U}_i\}_{i=1}^T$. Then, Fig. \ref{fig:pcs1} and Fig. \ref{fig:fcs} depict the preliminary and final candidate sets, respectively. Fig.~\ref{fig:pcs3} depicts the convergence analysis of the BO employed to jointly search over the reduced search space. Because the BO algorithm is implemented as a standard minimization problem, the objective is defined as the negative sum-rate, which is reflected on the vertical axis. Finally, Fig. \ref{fig:sc1} depicts the converged solution.

\subsection{Results across scenarios}
Fig. \ref{fig:allscenariosfina} presents the optimal RIS placement obtained for each scenario and Table \ref{table_results} lists the corresponding performance summary. From these results, we infer that even though Scenarios 1 and 2 have the same obstacle configuration, unlike a purely coverage based approach like \cite{loc1} and many others, our algorithm uses the user distribution to reduce the number of RISs needed from $2$ to $1$, to further optimize locations based on the same, and produces a similar sum rate and coverage performance. In Scenario 2 and 3, the non-homogeneity reduces the area we need to extend coverage to, reducing the number of RISs we need to achieve target coverage and throughput. On the other hand, Scenario 4 with a 2D Gaussian user distribution (with standard deviation of $15$~m) under a new obstacle configuration, requires $2$ RISs like Scenario 1 since the user density is non-zero throughout. Since the Gaussian density function concentrates users in a region that is well-served by the RISs, it results in the highest sum rate of $30.29$ bps/Hz and coverage of $99.28\%$ across 
all scenarios. Furthermore, for our default $p_{c_{th}} = 0.75$, we achieve a coverage of about $97\%$ across scenarios. 
\begin{figure}[t!]
\centering
     % \begin{subfigure}[b]{0.67\columnwidth}
        \includegraphics[width=0.63
        \linewidth]{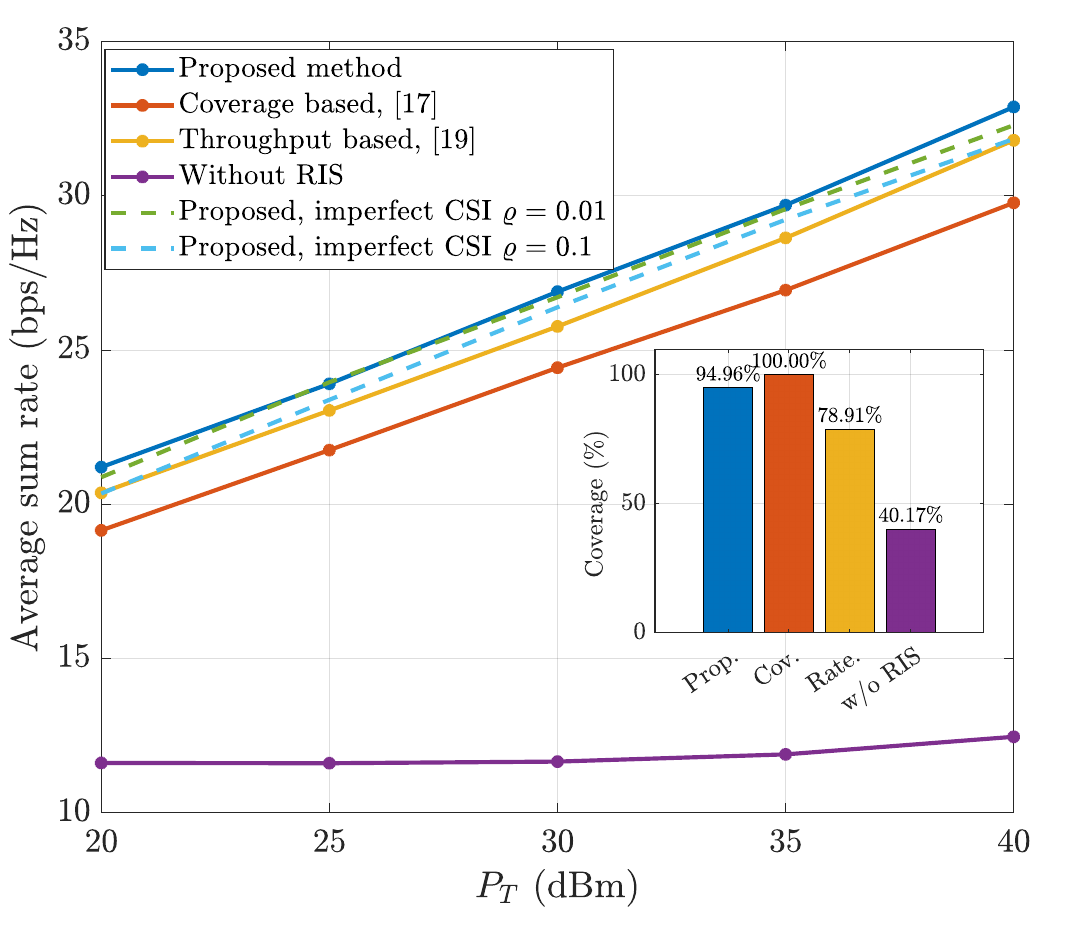}
    \caption{Conventional Performance analysis of proposed algorithm in Scenario 2.\label{fig:conven}}
  \end{figure} 
\begin{figure}[t!]
\centerline{\includegraphics[width=0.3\textwidth]{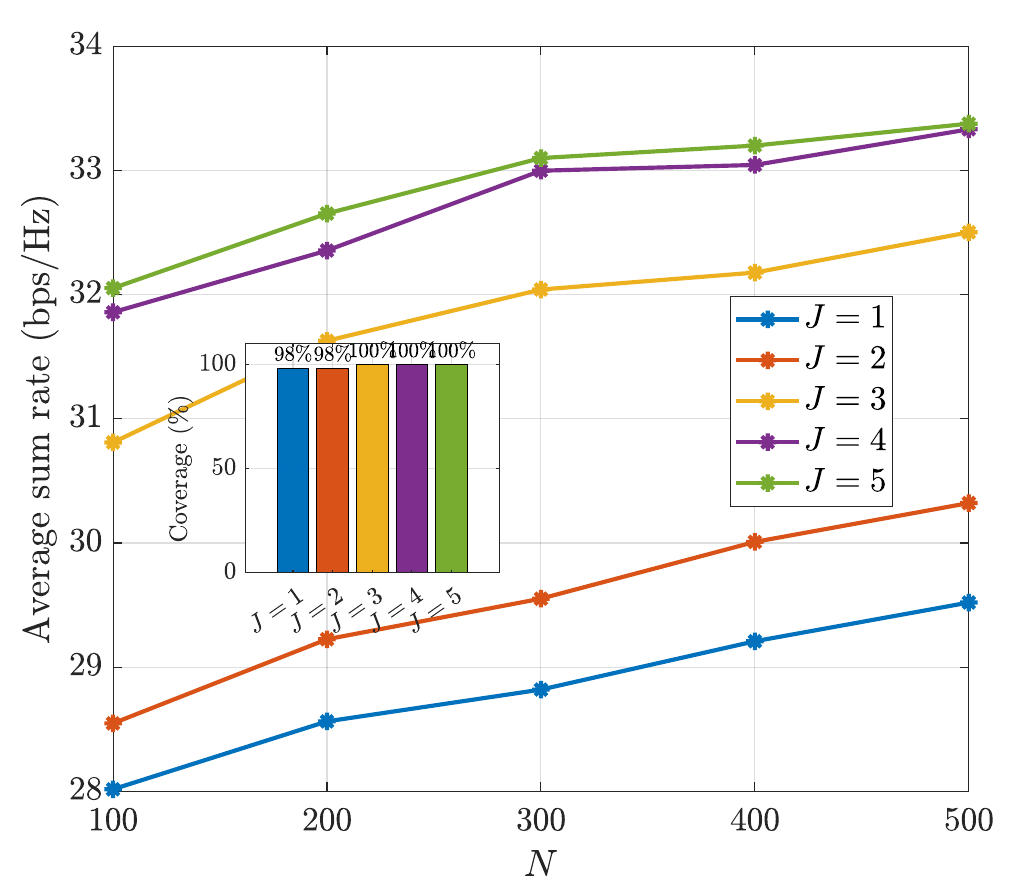}}
\caption{$J$ and $N$ characteristic for Scenario 2. \label{fig:JN}}
\end{figure} 
% \begin{figure}[t!]
% \centerline{\includegraphics[width=0.35\textwidth]{Figures/Results/RateRegion_cropped.pdf}}
% \caption{Pareto boundary for a fixed 2-user case in Scenario 2.\label{fig:rr}}
% \end{figure} 
\subsection{Conventional Communication Metrics}
\label{sec:conven}
For throughput analysis, we assume 3 baselines - (i) A greedy placement that maximizes coverage without a throughput guarantee \cite{loc1}, (ii) Sum rate maximization for placement with no coverage guarantee \cite{loc2}, and (iii) Without an RIS, with WMMSE beamforming. We discuss results for our implementation of these ideas in our simulation environment. More specifically, for \cite{loc1}, we manually search for the best RIS placement that maximizes coverage and add an additional RIS till $p_{c_{th}}$ is achieved. We implement \cite{loc2} by using a similar BO to search for the best RIS placement that maximizes throughput without any coverage constraint. For the $P_T$ characteristic, we implement these for a fixed number of BO iterations.

\subsubsection{$P_T$ characteristic} Average sum rate and Coverage results are discussed for Scenario 2 in Fig. \ref{fig:conven}. We see that even with WMMSE beamforming, there is a large gap between the throughput of the system without an RIS and the other cases. Though we use a similar simulation environment as~\cite{pathloss}, this gap is wider in our case because of the introduction of obstacles. Without an RIS, on average, only 40\% of users are covered. Only these users produce a non-zero SINR, reducing sum rate. We infer that maximizing coverage of an obstacle-deterred system, helps in maximizing its sum rate. However, if we adopt a purely coverage based approach, then we observe that even though 100\% coverage is achieved, a lesser throughput is observed. When we use a purely expected rate approach for RIS placement, the coverage achieved by the converged placement is $80\%$ -- a solid $20\%$ drop. This has an immediate consequence as argued before -- on average, $20\%$ are not served, resulting in negligible SINR in sum rate. As a result of this and owing to fixed number of BO iterations, it results in a lesser mean throughput than the proposed method. %, we observe that

Thus, from this analysis, we infer that generally, if a certain RIS configuration results in better coverage than another, it results in better throughput for an obstacle-deterred system. This pattern saturates at a certain coverage. Attempting to enhance coverage beyond it results in a decrease in throughput. %The proposed algorithm provides a coverage guarantee by restricting the search space, but also maximizes sum rate in that space to perform the best.

\subsubsection{$J$ vs $N$} While designing multi-RIS assisted systems, an important tradeoff to analyze is between the number of RISs and the number of reflecting elements in each of them\footnote{Range of $J$ is chosen as $[1,5]$ based on the considered scenario. Addition of more RISs as illustrated doesn't yield improvement in sum rate.} as in Fig. \ref{fig:JN}. We observe that across $J$, we see an improvement in throughput when we add more elements to each RIS. We observe that adding an additional RIS does not always result in an improvement in throughput. Even though $J=2$ has a similar coverage as $J=1$, we see that $J=2$ has a slightly improved throughput than $J=1$ case because the number of reflecting surfaces has doubled. On the other hand, we observe a considerable improvement in sum rate when a third RIS is added because we see that this RIS increases coverage to a full 100\%. This adds value to our earlier inference that coverage is tied to throughput. A slight increase in throughput is observed when increasing $J=3$ to $J=4$ like from $J=1$ to $J=2$ because of the increase in number of reflecting elements. Finally, we observe that sum rate saturates as $J$ is increased beyond $4$ because full coverage is already  attained. 
%emil1,
\begin{remark}
    Although ~\cite{pathloss,emil2} that prove that theoretically, $\gamma_k \propto$ $N^2$, ~\cite{pathloss} shows mathematically and analytically that the 3GPP Microcell pathloss environment (which we use) has strong direct path link, resulting in a gain of $3$ bps/Hz in an obstacle-free environment for $N\in[100,900]$. So, our achieved gain of about $1.8$ bps/Hz for $N\in[100,500]$ in an obstacle-filled environment is justified. Also, works like \cite{distrib1} compare a centralized deployment of RISs with a distributed one. Conclusions indicate that the distributed placement that maximizes coverage, also yields an improvement in throughput in an obstacle-deterred environment whereas a centralized deployment performs better in an obstacle-free environment. This validates our inferences that coverage and throughput are tied to each other in an obstacle filled environment.
\end{remark}

% \subsubsection{Pareto boundary for a fixed two-user case}
% Finally, to examine the multi-user tradeoff, we consider the pareto boundary plot for a fixed 2-user case (at ($95.6$ m,$45.6$ m) and ($102.0$ m, $16.4$ m). For a fair comparison, we assume simulation parameters to be assumed as in \cite{rateregion} ($M = 2$, $N= 200$, $N_0 = -103$ dBm, similar path loss exponents, and considering no obstacles). For this setup, we observe that our proposed algorithm clearly outperforms theirs with a considerable $8$ bps$^2$/Hz$^2$ improvement in area, further validating our approach to maximize the overall performance of the system.

\begin{figure}[t!]
\centering
     \begin{subfigure}[b]{0.71\columnwidth}
        \includegraphics[width=\linewidth]{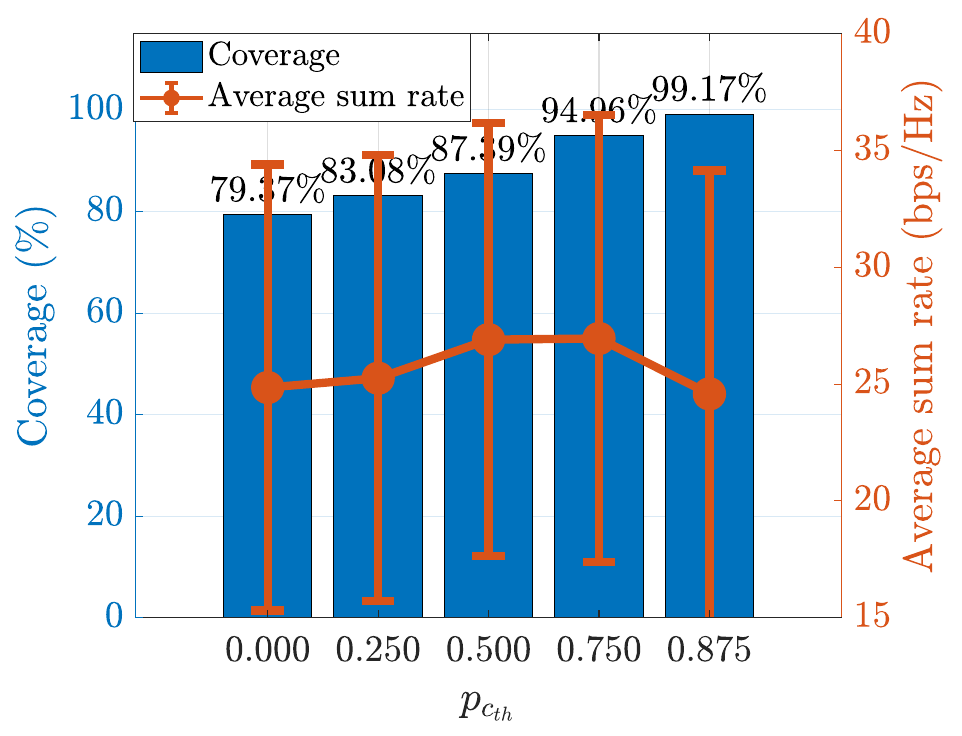}
        \caption{Average sum rate performance and Coverage.}
        \label{fig:pcth}
    \end{subfigure}
  
    \begin{subfigure}[b]{0.61\columnwidth}
        \includegraphics[width=\linewidth]{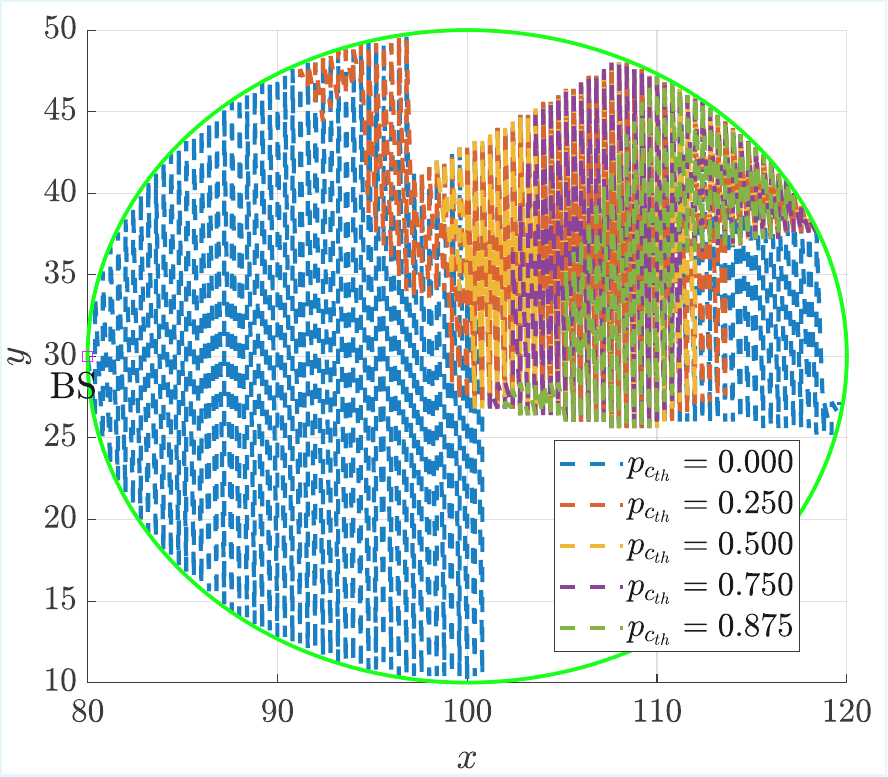}
        \caption{Search Space.}
        \label{fig:pcthsc}
    \end{subfigure}  
    \caption{Parametric Sweep of $p_{c_{th}}$ in Scenario 2.}
  \end{figure} 

\subsection{Parametric sweep of $p_{c_{th}}$}

Having established superior performance across baselines, we want to investigate the choice of an important parameter, $p_{c_{th}}$, for fixed number of BO iterations. In Fig. \ref{fig:pcthsc}, we observe that as we increase it from $0\%$ to $100\%$, the search space for RIS placement reduces from the set of all reachable points from the BS to a much smaller set. In Fig. \ref{fig:pcth}, as $p_{c_{th}}$ increases, the achieved coverage increases uniformly. On the other hand, as $p_{c_{th}}$ increases from $0\%$ to $50\%$, sum rate also increases due to two reasons again -- (i) no/low coverage assurance means more unserved users which in an obstacle deterred environment results in a decrease in average sum rate, (ii) no/low $p_{c_{th}}$ results in a larger search space for the BO. As we restrict the search space further to achieve a higher coverage, it leads to a clear decrease in sum rate. This observation aligns perfectly with the inferences and remarks made in Section \ref{sec:conven}. Thus, if we a good prior for the user process is available, the suggested choice of $p_{c_{th}}$ would be $50\%$ to $75\%$ -- which will assure an increase in coverage and most probably rate too. If reliable user prior isn't available, one can choose a uniform prior, lower $p_{c_{th}}$ and deploy the algorithm with increased BO iterations -- cannot effectively exploit user trends, but the larger search space may provide a generally effective solution. 

% , unless a specific application requires it to be more one-sided.

\section{Conclusion}
\label{sec:conc}

A novel, practically relevant problem formulation and an elegant solution were proposed. The proposed work tackled typical practical requirements of an RIS-assisted system, such as exploiting the user distribution to maximize performance, exploiting the obstacle configuration to maximize the coverage of the system, and handling user randomness. This is the first work, to the best of our knowledge, to handle the ideal hierarchical optimization directly -- maximizing throughput, while also ensuring good coverage. Also, in the proposed solution, a new and powerful approach that combined prior-driven learning with conventional beamforming to maximize expected post-deployment throughput was introduced. 

Results indicated an improvement in sum rate when choosing an RIS placement with improved coverage, but this trend reversed beyond a certain coverage. This behavior underscored a fundamental tradeoff in the design of an obstacle-deterred RIS-assisted system: a purely coverage-driven placement strategy narrows the search to regions that achieve near-complete coverage, but at the cost of placing RISs in locations that are suboptimal for signal enhancement. Conversely, a purely performance-driven strategy risks placing RISs at locations that maximize sum rate for only those users it can reach already. The proposed algorithm navigated this tradeoff by jointly accounting for both objectives, arriving at a placement that achieves strong coverage without sacrificing in sum rate gains. This balance between coverage and performance is what practical RIS deployment demands, and results validate that the proposed framework reliably achieves this balance across a variety of obstacle configurations and user distributions.

In future works, we can expand the idea for RIS-assisted ISAC systems to place RIS to maximize communication throughput and sensing accuracy.

% or


\appendices
\section{Proof of Theorem \ref{lm:cov}}
\label{app:lm1}

 \begin{IEEEproof}
% We want to prove that ensuring that $p_{c_{th}}$ percent coverage in each shadowed region $\mathcal{U}_i$ guarantees~\eqref{subeq:coverage} when the user positions follow the inhomogeneous process $\Psi_{u} = \{\boldsymbol{o}^{(k)}_u\}$ with the user distribution $f_u$.
Recall  $\mathcal{S} = \{\boldsymbol{s}_1, \boldsymbol{s}_2, \dots, \boldsymbol{s}_{N_{\text{grid}}}\}\subset \mathbb{R}^2$ where each $\boldsymbol{s}_i$ corresponds to a region $A_i \in \mathcal{O}$ such that the intensity is $\Lambda(\boldsymbol{s}_i) =\int_{A_i}\lambda_u(\boldsymbol{o})d\boldsymbol{o}$. Therefore, the pdf in~\eqref{eq:fu_def} reduces to a probability mass function (pmf) given by,
\begin{equation}
p_{\tilde{\boldsymbol{s}}}=\Pr(\boldsymbol{s}=\tilde{\boldsymbol{s}}) = \frac{\Lambda(\tilde{\boldsymbol{s}})}{\Lambda_u}, \tilde{\boldsymbol{s}} \in \mathcal{S}.
\end{equation}
Also, coverage probability~\eqref{eq:covdef} reduces to 
\begin{equation}
p_{cov}(\boldsymbol{O}_r) = \sum_{\tilde{\boldsymbol{s}}\in\mathcal{S}} \boldsymbol{1}_{\{L_{\mathrm{net}}(\tilde{\boldsymbol{s}}) \geq 1\}} \Pr(\boldsymbol{s}=\tilde{\boldsymbol{s}}),
\end{equation}
for an RIS deployment $\boldsymbol{O}_r$. From definitions, 
$\mathcal{S}=\mathcal{S}^*\cup\mathcal{S}^{\prime}$ is a disjoint partition, $\mathcal{U}=\mathcal{S}^*\setminus\mathcal{Q}\setminus\mathcal{V}$, and $\mathcal{U} = \bigcup_{i=1}^{T} \mathcal{U}_i$ is a disjoint partition. Let $\mathcal{V}^* = \mathcal{V}\setminus\mathcal{S}^{\prime}$ and $\mathcal{Q}^* = \mathcal{Q}\setminus\mathcal{S}^{\prime}$ respectively be the set of visible and obstacle points having a positive user intensity.  Since we model obstacles of solid circular and line types, users cannot be meaningfully present on/inside any of them, $\lambda_u(\boldsymbol{s})=0, \forall\boldsymbol{s}\in \mathcal{Q}$. Therefore, $\mathcal{Q}^* = \emptyset$. That yields a disjoint partition, $\mathcal{S} = \mathcal{S}^\prime\cup\mathcal{V}^* \cup \bigcup_{i=1}^{T} \mathcal{U}_i$%. Since  $\boldsymbol{1}_{\{L_{\mathrm{net}(\tilde{\boldsymbol{s}})\ge1}\}}$ and $\Pr(\boldsymbol{s}=\tilde{\boldsymbol{s}})$ are non-negative quantities, the sum over a subset of $\mathcal{S}$ yields the bound%We further partition $\mathcal{Q}^*$ by visibility -- $\mathcal{Q}_{\mathcal{V}}^*=\mathcal{Q}^*\cap\mathcal{V}$ and $\mathcal{Q}_{\mathcal{V}^c}^*=\mathcal{Q}^*\cap(\mathcal{S}\setminus\mathcal{V})$.
\begin{align}
p_{cov}(\boldsymbol{O}_r) =\sum_{i\in\{\mathcal{S}^\prime,\mathcal{V}^*,\mathcal{U}_1,\dots\mathcal{U}_T\}}\sum_{\tilde{\boldsymbol{s}}\in i} \boldsymbol{1}_{\{L_{\mathrm{net}(\tilde{\boldsymbol{s}})} \geq 1\}} \Pr(\boldsymbol{s}=\tilde{\boldsymbol{s}}).
% \\
% &+ \int_{\mathcal{S^{\prime}}} \boldsymbol{1}_{\{L_{\mathrm{net}(\boldsymbol{s})} \geq 1\}} f_u(\boldsymbol{s}) d\boldsymbol{s}\\
% &+\int_{\mathcal{V}^*} \boldsymbol{1}_{\{L_{\mathrm{net}(\boldsymbol{s})} \geq 1\}} f_u(\boldsymbol{s}) d\boldsymbol{s}\\ &+ \sum_{i=1}^T \int_{\mathcal{U}_i} \boldsymbol{1}_{\{L_{\mathrm{net}}(\boldsymbol{s}) \geq 1\}} f_u(\boldsymbol{s}) d\boldsymbol{s}.
\end{align}
%Since it is a non-disjoint union ($\tilde{\mathcal{Q}}\cap\mathcal{S}\ne \emptyset$), we use the union bound. Also, by definition, for all $\boldsymbol{s}_1 \in \mathcal{V}^*, L(\boldsymbol{s}_b, \boldsymbol{s}_1) = 1$ and for all $\boldsymbol{s}_2 \in \tilde{\mathcal{Q}}, L(\boldsymbol{s}_b, \boldsymbol{s}_2) = 0$. 
Manipulating~\eqref{eq:local} further we get \begin{align}
&\frac{\frac{1}{\Lambda_u}\sum_{\tilde{\boldsymbol{s}}\in\mathcal{U}_i} \boldsymbol{1}_{\{L_{\mathrm{net}(\tilde{\boldsymbol{s}})} \geq 1\}} \Lambda(\tilde{\boldsymbol{s}})}{\frac{1}{\Lambda_u}\sum_{\tilde{\boldsymbol{s}}\in\mathcal{U}_i}\Lambda(\tilde{\boldsymbol{s}})}  \geq p_{c_{\mathrm{th}}},\\
&\frac{\sum_{\tilde{\boldsymbol{s}}\in\mathcal{U}_i} \boldsymbol{1}_{\{L_{\mathrm{net}(\tilde{\boldsymbol{s}})} \geq 1\}} \Pr(\boldsymbol{s}=\tilde{\boldsymbol{s}})}{\sum_{\tilde{\boldsymbol{s}}\in\mathcal{U}_i}\Pr(\boldsymbol{s}=\tilde{\boldsymbol{s}})}  \geq p_{c_{\mathrm{th}}}.
\end{align}
By definition, $\boldsymbol{1}_{\{L_{\mathrm{net}(\tilde{\boldsymbol{s}})} \geq 1\}} = 1$ for all $\tilde{\boldsymbol{s}} \in \mathcal{V}^*$ and $\Pr(\boldsymbol{s}=\tilde{\boldsymbol{s}}) = \Lambda(\tilde{\boldsymbol{s}})/\Lambda_u=0$ for all $\tilde{\boldsymbol{s}} \in \mathcal{S}^\prime$. Using these we have
\begin{equation}
p_{cov}(\boldsymbol{O}_r) \geq \sum_{\tilde{\boldsymbol{s}}\in\mathcal{V}^*} \Pr(\boldsymbol{s}=\tilde{\boldsymbol{s}}) + p_{c_{\mathrm{th}}} \sum_{i=1}^T  \sum_{\tilde{\boldsymbol{s}}\in\mathcal{U}_i} \Pr(\boldsymbol{s}=\tilde{\boldsymbol{s}}).
\end{equation}
% Consider $p_{\mathcal{V}^*} = \int_{\mathcal{V}^*} f_u(\boldsymbol{s}) d\boldsymbol{s}$, $p_{\mathcal{U}_i} = \int_{\mathcal{U}_i} f_u(\boldsymbol{s}) d\boldsymbol{s}$, $p_{\mathcal{S}^*} = \int_{\mathcal{S}^*} f_u(\boldsymbol{s}) d\boldsymbol{s}$ and $p_{\mathcal{S}} = \int_{\mathcal{S}} f_u(\boldsymbol{s}) d\boldsymbol{s} = 1$. Since $\mathcal{V}^* \cup (\bigcup_{i=1}^{T} \mathcal{U}_i) = \mathcal{S}^*$ is a disjoint union,
% we have $p_{\mathcal{S}^*}=p_{\mathcal{V}^*}+\sum_{i=1}^{T}p_{\mathcal{U}^*}$. Recall that $\mathcal{S}^*$ is the set of all points that have positive user intensity. Since the user intensities $\lambda_u(\boldsymbol{s})$ are non-negative at all grid points $\boldsymbol{s}$ (which is the meaningful case) and using~\eqref{eq:fu_def}, we get that $p_{\mathcal{S}^*}=p_{\mathcal{S}}=1$. Therefore,
\begin{align}
&\Rightarrow p_{cov}(\boldsymbol{O}_r) \geq \Pr(\tilde{\boldsymbol{s}}\in\mathcal{V}^*) + p_{c_{\mathrm{th}}} \Pr(\tilde{\boldsymbol{s}}\in\mathcal{U})
%(1 - p_{\mathcal{V}^*}),
%&\Rightarrow p_{cov}(\boldsymbol{O}_r) \geq.
\end{align}
Since $\mathcal{S}$ is our universal set, i.e., $\Pr(\tilde{\boldsymbol{s}}\in\mathcal{S})=1$ and $\Pr(\tilde{\boldsymbol{s}}\in\mathcal{S}^\prime)=0$ by definition, we have
\begin{align}
p_{cov}(\boldsymbol{O}_r) &\geq \Pr(\tilde{\boldsymbol{s}}\in\mathcal{V}^*) + p_{c_{\mathrm{th}}} (1-\Pr(\tilde{\boldsymbol{s}}\in\mathcal{V}^*))\\
&=p_{c_{\mathrm{th}}}+\Pr(\tilde{\boldsymbol{s}}\in\mathcal{V}^*)(1-p_{c_{\mathrm{th}}})
%(1 - p_{\mathcal{V}^*}),
%&\Rightarrow p_{cov}(\boldsymbol{O}_r) \geq.
\end{align}
Since $0 \le p_{c_{\mathrm{th}}},\Pr(\tilde{\boldsymbol{s}}\in\mathcal{V}^*) \leq 1$, $p_{cov}(\boldsymbol{O}_r) \geq p_{c_{\mathrm{th}}}$ follows.
\end{IEEEproof}

\section{Proof of Proposition \ref{prop:1}}
\label{app:prop1}
\begin{IEEEproof}
We consider multiple RISs that provide independent single hop links for the downlink transmission to users. Also recall $\mathcal{U}=\mathcal{S}^*\setminus\mathcal{Q}\setminus\mathcal{V} = \bigcup_{i=1}^T\mathcal{U}_i$ only has points that do not have a direct link from the BS, but has a chance of user presence. So, if each potential RIS location is guaranteed to have a direct path to the BS, i.e., $\boldsymbol{s}\in\mathcal{V}$ and if this potential RIS point has a direct path to a potential user $\boldsymbol{u}$, i.e., $(\boldsymbol{s},\boldsymbol{u})\in E_{\mathrm{vis}}$, then $\boldsymbol{u}$ has at least one path from BS, satisfying $L_{\mathrm{net}(\tilde{\boldsymbol{s}})}\ge1$
in~\eqref{eq:local}. Since~\eqref{eq:local} was specifically derived to assure $p_{c_{th}}$ percent coverage in each $\mathcal{U}_i$, we directly  get~\eqref{eq:candidate_set} to be the corresponding RIS candidate set where if an RIS is positioned in one of those locations, it is guaranteed to cover $p_{c_{th}}$ percent of each  $\mathcal{U}_i$. From Theorem \ref{lm:cov}, since $p_{c_{th}}$ percent coverage in each $\mathcal{U}_i,\text{ } 1\le i\le T$ is assured,~\eqref{subeq:coverage} is satisfied.\end{IEEEproof}

\section{Proof of Proposition \ref{prop:2}}
\label{app:prop2}
\begin{IEEEproof}
Let $\mathcal{P} = \{e_1,\ldots,e_J\}$ be a valid partition on $V_{\mathrm{int}}$. A hyperedge $e_j =(C^\prime_{j_1}, \dots, C^\prime_{j_f}) \in \mathcal{P}$ indicates a non-empty $C_j=\bigcap_{i=j_1}^{j_f}C^{\prime}_i\ne \emptyset$. Suppose we deploy an RIS at a $\boldsymbol{o}^{(j)}_r\in C_j$. Since $\boldsymbol{o}^{(j)}_r\in C^\prime_i,i\in\{j_1,\dots j_f\}$, by Proposition \ref{prop:1}, it assures $p_{c_{th}}$ percent coverage in each $\{\mathcal{U}_i\}_{i=j_1}^{j_f}$. Since $\bigcup_{i=1}^Je_i=V_{\mathrm{int}}=\{C^\prime_i\}_{i=1}^T$, deploying $J$ RISs in each $\{C_j\}_{j=1}^J$ constructed this way assures~\eqref{eq:local} for all $\{\mathcal{U}_i\}_{i=1}^{T}$, and therefore by Theorem~\ref{lm:cov} satisfies~\eqref{subeq:coverage}. Since every point $\boldsymbol{o}^{(i)}_r\in C^\prime_i\text{ }\forall 1 \le i \le T$ lies in $\mathcal{V}$ by Proposition~\ref{prop:1}, every $\boldsymbol{o}^{(j)}_r\in C_j\text{ }\forall 1 \le j \le J$ also lies in $\mathcal{V}$, satisfying~\eqref{subeq:BS}. Thus, for every valid partition $\mathcal{P}$, we have an $\boldsymbol{O}_r=[\boldsymbol{o}^{(1)}_r,\dots,\boldsymbol{o}^{(J)}_r]\text{ s.t. }\boldsymbol{o}^{(j)}_r\in C_j\text{ }\forall 1 \le j \le J$ that lies in the constraint space of $\mathbf{P1}$. To minimize $J=\mathrm{cols}(\boldsymbol{O_r})$, we  compute that partition with least cardinality in Proposition~\ref{prop:2}.\end{IEEEproof}
% From Theorem \ref{prop:1}, each preliminary candidate set $C'_{i}$ assures $p_{c_{th}}$ percent of $\mathcal{U}_i$ is covered. Thus, a hyperedge $(C'_{j_1}, \dots, C'_{j_f})$ signifies the existence of $\bigcap_{i=j_1}^{j_f}C^{\prime}_i$ that simultaneously covers $p_{c_{th}}$ percent of all of $\mathcal{U}_{j_1}, \dots, \mathcal{U}_{j_f}$. Therefore, we define
% \begin{equation}
%     C_i = \bigcap_{j\in e_i}C^{\prime}_j, \quad 1\le j\le J,
% \end{equation}
% to be the final collection of candidate RIS sets that can assure $p_{c_{th}}$ percent for all $\{\mathcal{U}_i\}_{i=1}^T$. From $\mathbf{P1}$, we need to ensure $J$ is minimized within this constraint space. Mathematically, we want to minimize the partition $\mathcal{P} = \{e_1,\ldots,e_J\}$ such that $\bigcup_{i=1}^Je_i=V_{\mathrm{int}}$. This problem is defined by 
% \begin{align}
%     &\min_{\mathcal{P}} |\mathcal{P}|\\
%     \text{s.t.}\quad & e_i \in E_{\mathrm{int}}, \forall e_i \in \mathcal{P}.
% \end{align}
% Therefore, if we partition the hypergraph into a set of hyperedges such that cardinality of the partition is the minimized, we get the least possible number of RISs we can deploy that satisfies $p_{c_{th}}$ percent coverage for all $T$ shadowed regions. As an extension, we get each reduced candidate set to be the intersection of all vertices of the corresponding hyperedge.

% \section*{Acknowledgment}


\ifCLASSOPTIONcaptionsoff
  \newpage
\fi



\bibliographystyle{IEEEtran}
\bibliography{another}

\end{document}